\documentclass[a4paper,USenglish]{lipics-v2021}
\usepackage{subcaption}
\usepackage{tikz}
\usepackage{nicefrac}
\usepackage{appendix}
\usepackage{todonotes}
\usepackage[most]{tcolorbox}
\usepackage{algorithm}
\usepackage{algpseudocode}

\newcommand{\cort}{\mathsf{courteous}}
\newcommand{\court}{\cort}
\newcommand{\pref}{\mathsf{pref}}

\newcommand{\sweep}{\mathsf{sweep}}
\newcommand{\cut}{\mathsf{cut}}
\newcommand{\col}{\mathsf{col}}
\newcommand{\alg}{\mathtt{ALG}}

\def\In{\mathrm{input}}
\def\Out{\mathrm{output}}

\usepackage{amsmath}
\usepackage{graphicx}
\usepackage{wrapfig}
\usepackage{xspace}
\usepackage{url}
\usepackage{color,soul}

\newcommand{\namedref}[2]{\hyperref[#2]{#1~\ref*{#2}}}
\newcommand{\sectionref}[1]{\namedref{Section}{#1}}

\newcommand{\theoremref}[1]{\namedref{Theorem}{#1}}

\newcommand{\figref}[1]{\namedref{Figure}{#1}}
\newcommand{\lemmaref}[1]{\namedref{Lemma}{#1}}
\newcommand{\tableref}[1]{\namedref{Table}{#1}}

\newcommand{\appref}[1]{\namedref{Appendix}{#1}}

\newcommand{\algoref}[1]{\namedref{Algorithm}{#1}}
\newcommand{\linref}[1]{\namedref{Line}{#1}}
\newcommand{\equalityref}[1]{\hyperref[#1]{Eq.~\ref*{#1}}}
\newcommand{\inequalityref}[1]{\hyperref[#1]{Inequality~(\ref*{#1})}}

\newcommand{\stepref}[1]{\hyperref[#1]{Step~(\ref*{#1})}}

\newcommand{\Set}[1]{\left\{ #1 \right\}}

\newcommand{\ceil}[1]{\left\lceil #1 \right\rceil}
\newcommand{\floor}[1]{\left\lfloor #1 \right\rfloor}
\newcommand{\ignore}[1]{}

\def\DEF{\stackrel{\mathrm{def}}{=}}

\def\Nat{\mathbb{N}}

\newcommand{\Eqr}[1]{Eq.~(\ref{#1})}
\def\mtoday{\ifcase\month\or
  January\or February\or March\or April\or May\or June\or July\or
  August\or September\or October\or November\or December\fi\space\number\year}

\title{\bf Consensus with Stochastic Broadcast}
\author{Pierre Fraigniaud}
{Institut de Recherche en Informatique Fondamentale (IRIF)\\ CNRS and Université Paris Cité, France}
{pierre.fraigniaud@irif.fr}
{https://orcid.org/0000-0003-4534-4803}
{Additional support from ANR Projects ENEDISC (ANR-24-CE48-7768-01) and PREDICTIONS (ANR-23-CE48-0010), and Émergence Project METALG (U. Paris Cité).}

\author{Boaz Patt-Shamir}
{Tel Aviv University, Israel. }
{boaz@tau.ac.il}
{http://orcid.org/0000-0001-8398-8218}
{Work partially done while visiting IRIF at Université Paris Cité, supported in part by Israel Science Foundation grant 1948/21, and Fondation des Sciences Mathématiques de Paris (FSMP).}

\author{Sergio Rajsbaum}
{Institut de Recherche en Informatique Fondamentale (IRIF)\\ CNRS and Université Paris Cité, France}
{sergio.rajsbaum@gmail.com}
{https://orcid.org/0000-0002-0009-5287}
{Support from  project PAPIIT 37-IN109926. On leave from Instituto de Matem\'aticas, Universidad Nacional Aut\'onoma de M\'exico.}

\authorrunning{P. Fraigniaud, B. Patt-Shamir and S. Rajsbaum}
\Copyright{Pierre Fraigniaud, Boaz Patt-Shamir and Sergio Rajsbaum}

\ccsdesc[500]{Theory of computation~Probabilistic computation}
\ccsdesc[500]{Theory of computation~Distributed computing models}
\ccsdesc[500]{Theory of computation~Random network models}
\ccsdesc[500]{Theory of computation~Distributed algorithms}

\keywords{stochastic communication, distributed agreement, distributed consensus, Kripke graph, topology of distributed algorithms}
\EventEditors{Ioannis Chatzigiannakis, Andrea Vitaletti, Keren Censor-Hillel, and William K. Moses Jr.}

\nolinenumbers
\hideLIPIcs

\begin{document}
%%%%%%%%%%%%%%%%%%%%%%%%%%%%%%%%%

\maketitle

\begin{abstract}
 We study binary consensus in the \emph{stochastic broadcast model},
which assumes $n\geq 2$ processes communicating synchronously by
\emph{message broadcasts}. At each round, every process broadcasts a
message to all the other processes. Each broadcast succeeds
independently with some probability $p$, where $p\in[0,1]$ is a
parameter of the model.  If a broadcast succeeds, all
processes receive the message, and if it fails, no process receives
the message.  The sender does not know whether its broadcast was
successful or not.  In this model, consensus is not solvable; the
objective is to design, for a given number of rounds $r$, consensus
algorithms that terminate in $r$ rounds, minimizing the probability of
error (disagreement).  This problem has been studied in depth for 2
processes [DISC 2025].  We
extend the study to $n> 2$.

We identify two
key algorithms in the $1$-round setting.  In Algorithm $\pref1$, each
process decides $1$ unless it knows of no input $1$. In Algorithm
$\court$, each process decides the more frequent value in its view,
and, in case of a tie, the process decides the negation of its own
input value.  It turns out that for small values of $p$ $\court$
is better, while for larger values, $\pref1$ is better.

Our main results are lower bounds.  For $n=3$
processes and one round, we show that $\court$ is optimal for
$p\in[0,\nicefrac23$], and $\pref1$ is optimal for
$p\in[\nicefrac23,1]$.  To obtain this result, we develop a specialized
theory that extends the combinatorial topology approach to a
probabilistic setting.
    
For an arbitrary (but fixed) number $n>3$ of processes,    we show
that  among all 1-round consensus algorithms, $\court$ is optimal for
$p\in[0,1/n]$, and that  $\pref1$ is asymptotically optimal
for  $p\to1$.  
    
For multi-round algorithms, we show that (unlike the case of $n=2$
processes) repeating $r$ times an optimal 1-round algorithm is \emph{not} 
optimal for $n\geq 3$ processes. Nevertheless, we present a 2-round
algorithm, called $\sweep$, that guarantees the best error probability
for the number of \emph{transmissions}. Specifically, for every number
$n$ of processes, repeating $\sweep$ $r$ times 
entails at most $rn$ broadcasts, and results in error probability
$q^{rn}$, where $q=1-p$ is the probability that a broadcast fails.

We believe that our results provide another step towards the goal of a general topological theory of randomized distributed computing.
\end{abstract}

%%%%%%%%%%%%%%%%%%%%%%%%%%%%%%%%%
\section{Introduction}
%%%%%%%%%%%%%%%%%%%%%%%%%%%%%%%%%

We  consider  the \emph{stochastic broadcast model} (see, e.g.,~\cite{PelcPeleg2005}), consisting of  a set of $n$ synchronous processes communicating 
over broadcast channels as follows.
In each round, every process can broadcast a message to all  other processes. If a broadcast is successful, the message is delivered  to
all processes (at the same round); and if the
broadcast fails, no process receives the message.
Message delivery is stochastic, 
parametrized by  $p\in [0,1]$: each broadcast, independently, succeeds with probability $p$, or fails with probability $q= 1-p$. 
Crucially, the sender does not know whether its broadcast was successful or not (otherwise consensus is trivial). 

In this {stochastic} framework, we study binary consensus by \emph{deterministic} algorithms that terminate in a fixed number of rounds. 
It is straightforward that consensus cannot be solved, even if allowing an arbitrarily large number of rounds.
The goal is to minimize the probability of error,  
over all possible inputs, for a given number of rounds.
More precisely, 
processes start with binary input values and decide  binary values.
A deterministic $r$-round algorithm is said to solve \emph{consensus with error probability} $\varrho\in[0,1]$ if the following conditions are satisfied.
\begin{description}
\item[Termination:] All processes terminate in $r$ rounds, and decide  binary output values.
\item[Validity:] If all processes start with the same binary value, they must all decide on this value.
\item[Randomized agreement:] For every  assignment of binary inputs to the processes, the probability that some process decides $0$
 while some other process decides $1$ is at most $\varrho$.
\end{description}
Note that the error probability $\varrho$ is due to the nature of the stochastic environment: in this paper, we consider only deterministic algorithms. 

For two processes, Fraigniaud, Patt-Shamir, and Rajsbaum~\cite{FraigniaudPSR25} studied \emph{approximate agreement} in the stochastic broadcast model and showed that, in this model, it is intimately related to consensus. 
Specifically, they showed that,  
to obtain optimal solutions for consensus, i.e., algorithms minimizing the probability of disagreement, it is necessary to use different algorithms for different ranges of~$p$.
\begin{itemize}
    \item For $p\geq 1/2$, one optimal algorithm (called AMP$(1)$ in~\cite{FraigniaudPSR25}) says that each process decides~$1$ unless no input~$1$ is known to the process. 

    \item For $p\leq 1/2$, another rule (called FV in~\cite{FraigniaudPSR25}) is optimal: each process decides the input value of the \emph{other} process, unless it knows only its own input value. 
\end{itemize}
It was shown that, for all $r\in \mathbb N$, 2-process consensus can be solved in $r$ rounds with error probability at most $\left(\min\{p^2+q^2,q\}\right)^r$ by recursively repeating these algorithms for $r$ rounds, and that no protocol can solve consensus in $r$ rounds with error probability smaller than $\left(\min\{p^2+q^2,q\}\right)^r$.

However, it was also shown in~\cite{FraigniaudPSR25} that, for $n\geq 3$ processes, consensus is no longer directly related to approximate agreement. That is, while for two processes it is possible to solve approximate agreement optimally by having each process decide an integral value (0 or 1), this is not the case for three processes, i.e., solving approximate agreement optimally in the stochastic broadcast model requires processes to decide fractional values in $(0,1)$. The objective of this paper is to extend the results about consensus in~\cite{FraigniaudPSR25} to the case of  $n\geq 3$ processes. This requires developing an approach different from the one followed in~\cite{FraigniaudPSR25} for two processes, as fractional outputs are not valid solutions for consensus.  

%%%%%%%%%%%%%%%%%%%%%%%%%%%%%%%%%
\subsection{Our Results}
%%%%%%%%%%%%%%%%%%%%%%%%%%%%%%%%%

In summary, it turns out that the probability of error of a one-round, $n$ process consensus algorithm is a polynomial in $p$ of degree $n$.
Different algorithms may have different polynomials, but always when $p=0$ they have error probability of $1$.
Remarkably, there is no single algorithm that is optimal on all the range $0\leq p\leq 1$. 
It is necessary to use different algorithms for different ranges of $p$. Thus,
the optimal one-round error probability is \emph{not} described by a polynomial over all the range $[0,1]$.
The multi-round case has an even more complex behavior. 
In more detail, here follow our results.

\subparagraph{Two  one-round algorithms.}

We identify two  one-round algorithms, named $\cort$ and $\pref1$, which deserve special attention. 
Algorithm $\pref1$ is a generalization of Algorithm AMP$(1)$ from~\cite{FraigniaudPSR25}: each process decides~$1$ unless it knows of no input~$1$. The other algorithm, $\cort$, generalizes Algorithm FV of~\cite{FraigniaudPSR25} to $n>2$ processes: each process decides the most frequent value in its view, %(i.e., the multiset of received values, including its own input value), 
and, in case of a tie, the process is ``courteous'' in the sense that it decides the other value, i.e., the negation of its own input value. 
We compute their probability of error (a polynomial of degree $n$ in $p$, but the formulas are often more concise or informative  if we expressed them in terms of both $p$ and $q=1-p$): 
for $\pref1$ the error probability is linear, simply $q$, while for $\court$  it is 
    $$
    \sum_{i=0}^{\floor{n/2}}
    \binom{\floor{n/2}}{i}\binom{\ceil{n/2}}{i}p^{2i}q^{n-2i}~.
  $$
Instantiating for $n=3$ processes, $\cort$ and $\pref1$ have respective error probabilities  $2p^2q+q^3$  and~$q$ . For $p\leq\nicefrac23$,  $\cort$ offers a lower error probability than $\pref1$, whereas $\pref1$ offers a lower error probability than $\cort$ for $p\geq\nicefrac23$. At  $p= \nicefrac23$ their error probabilities coincide.
The graphs are in~\figref{fig:optAlgs}, comparing the results with the  case of $n=2$, where the crossover is at $p=\nicefrac12$.

\begin{figure}[ht]
\centering
\includegraphics[scale=.4]{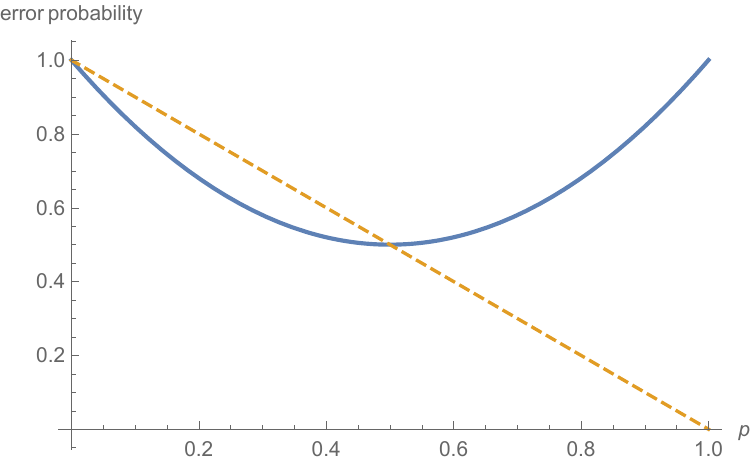}
\hspace*{6mm} \includegraphics[scale=.4]{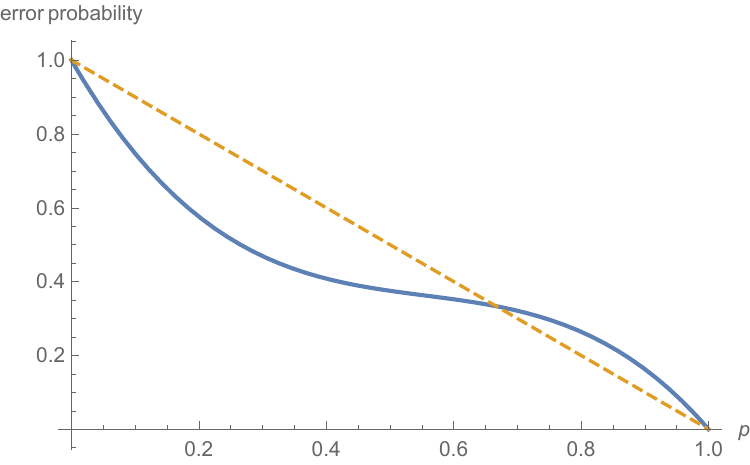}
\caption{The error probability of the $2$- and $3$-player (left and right figures, repspectively) $1$-round algorithms $\cort$ and $\pref1$ as a
function of~$p$. The blue line is for $\cort$, and the
orange (dashed) line is for $\pref1$. The intersection shifts from $p=\nicefrac12$ when $n=2$ to  $p=\nicefrac23$ when $n=3$.
 For both $n=2$ and $n=3$, each algorithm is optimal in the corresponding range.}
\label{fig:optAlgs}
\end{figure} 

\subparagraph{Lower bounds.}

For $n=3$, we extend the topological approach to distributed computing to the stochastic setting. This allows us to show that 
the combined use of $\pref1$ and $\cort$ gives an optimal one-round algorithm.
For every $p\in[0,1]$, the minimum error probability of any deterministic 3-process 1-round algorithm solving consensus is $2p^2q+q^3$ if $p\leq \nicefrac23$, and $q$ otherwise. 

For $n>3$ processes, 
we prove that $\cort$ is asymptotically optimal for small $p$ among all 1-round consensus algorithms for $n$ processes, whereas  $\pref1$ is asymptotically optimal for large $p$. That is, we show that, for every $n\geq 3$, there exists $p_{small}$ and $p_{large}$ such that, $\cort$ is optimal for all $p\leq p_{small}$, and $\pref1$ is optimal for $p\geq p_{large}$. In the case $n=3$, $p_{small}=p_{large}=\nicefrac23$.

\subparagraph{Multi-round algorithms.}

The error probability  goes down exponentially by simply repeating $r$ rounds
a one-round algorithm. For the case of $n=2$, using an optimal one-round algorithm gives an optimal $r$ round algorithm.
Surprisingly, this is not the case for $n>2$, although the error probability still goes down exponentially with $r$.
We show that repeating $r>1$ times an optimal 1-round algorithm is \emph{not} optimal for $n\geq 3$ processes (in contrast to the case of $n=2$). 
Nevertheless, we design a $2$-round algorithm, called  $\sweep$, that guarantees the best error probability for a given number of transmissions. 
Specifically, for every number $n$ of processes, repeating  $\sweep$ $r$ times (in $2r$ rounds) results in $rn$ broadcast transmissions, and (communication-optimal) error probability $q^{rn}$. 

%%%%%%%%%%%%%%%%%%%%%%%%%%%%%%%%%
\subsection{Our Techniques}  
\label{subsec:our-techniques}
%%%%%%%%%%%%%%%%%%%%%%%%%%%%%%%%%

Similarly to~\cite{FraigniaudPSR25}, 
our goal is not to optimize  asymptotic performance; rather, we are interested in optimal $r$-round $n$-process algorithms for \emph{fixed} values of~$n$ and~$r$. There are several reasons for this focus. First, we are interested in understanding the impact of stochastic failures in concrete settings in which a potentially small number of processes interact (like in a multi-core architecture) to solve consensus in a few rounds. Second, asymptotic results typically ignore constant factors, and we seek exact solutions.
Finally, at the high level, our overarching goal is to extend the topological theory of distributed computability to the stochastic case.
We briefly describe this extension of the topological theory of distributed computability here, a more detailed
presentation is in  Appendix~\ref{sec:Protocol-Complex}. This extension is the key to
showing that our one-round algorithms are optimal within the respective range of  $p\in[0,1]$.
 
The starting point is to consider the so-called \emph{protocol complex} $\mathcal{P}$ as usual in the topology perspective~\cite{HerlihyKR13}, focusing on the case of one round protocols. 
Recall that every vertex of $\mathcal{P}$ is a pair $(i,w_i)$ where $i\in[n]$ is a process name, and $w_i$ is a view of process~$i$ after one round.  
A set $\{(1,w_1),\dots,(n,w_n)\}$ is a simplex (actually, a facet) whenever the $n$ views $w_1,\dots,w_n$ of the $n$ processes are compatible after one round of communication. 
See Fig.~\ref{fig:simplex-flat-proba-mainText} for the subcomplex of the protocol complex induced by a fixed input,
for $n=2,3$.

\begin{figure}[hb]
\centering
{\includegraphics[scale=.4]{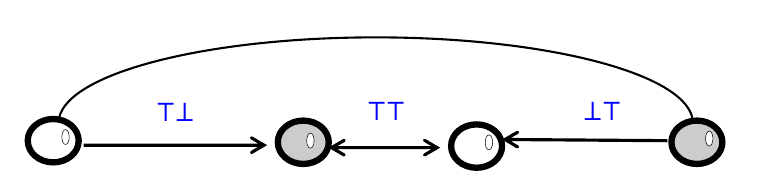}}
\hspace*{1cm}
{\includegraphics[scale=.4]{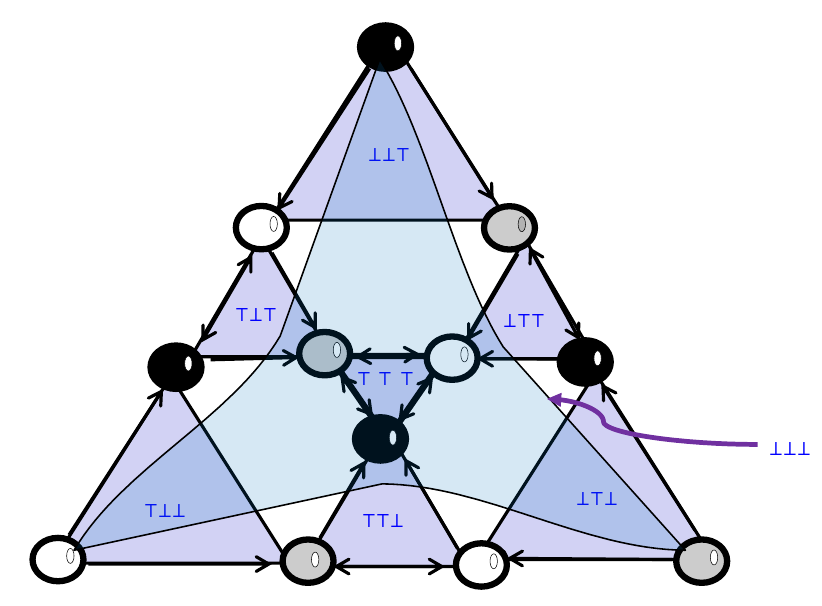}}
\caption{The 1-round protocol complex for  $n=2,3$, starting with a fixed input vector. Process 1 is colored white, 2 is gray and 3 is black. Arrows indicate successful message deliveries (they are not part of the complex). 
}
\label{fig:simplex-flat-proba-mainText}
\end{figure}

A consensus algorithm defines a binary coloring of the vertices of the protocol complex. This is because
the decisions of the processes are a function of their views, which correspond to the vertices of the protocol 
complex.
See Fig.~\ref{fig:optAlgsComplex-mainText}. 
Furthermore, 
 the simplices for which an algorithm fails  induce a cut in the protocol complex, as illustrated in the figure.
 Instead of trying to find the minimum ``simplicial-cut'' directly, we work with the dual problem, described next.
Notice that finding cuts  in simplicial complexes is NP-hard~\cite{MaxwellN21}. 

\begin{figure}[tb]
  \centering
{\includegraphics[scale=.4]{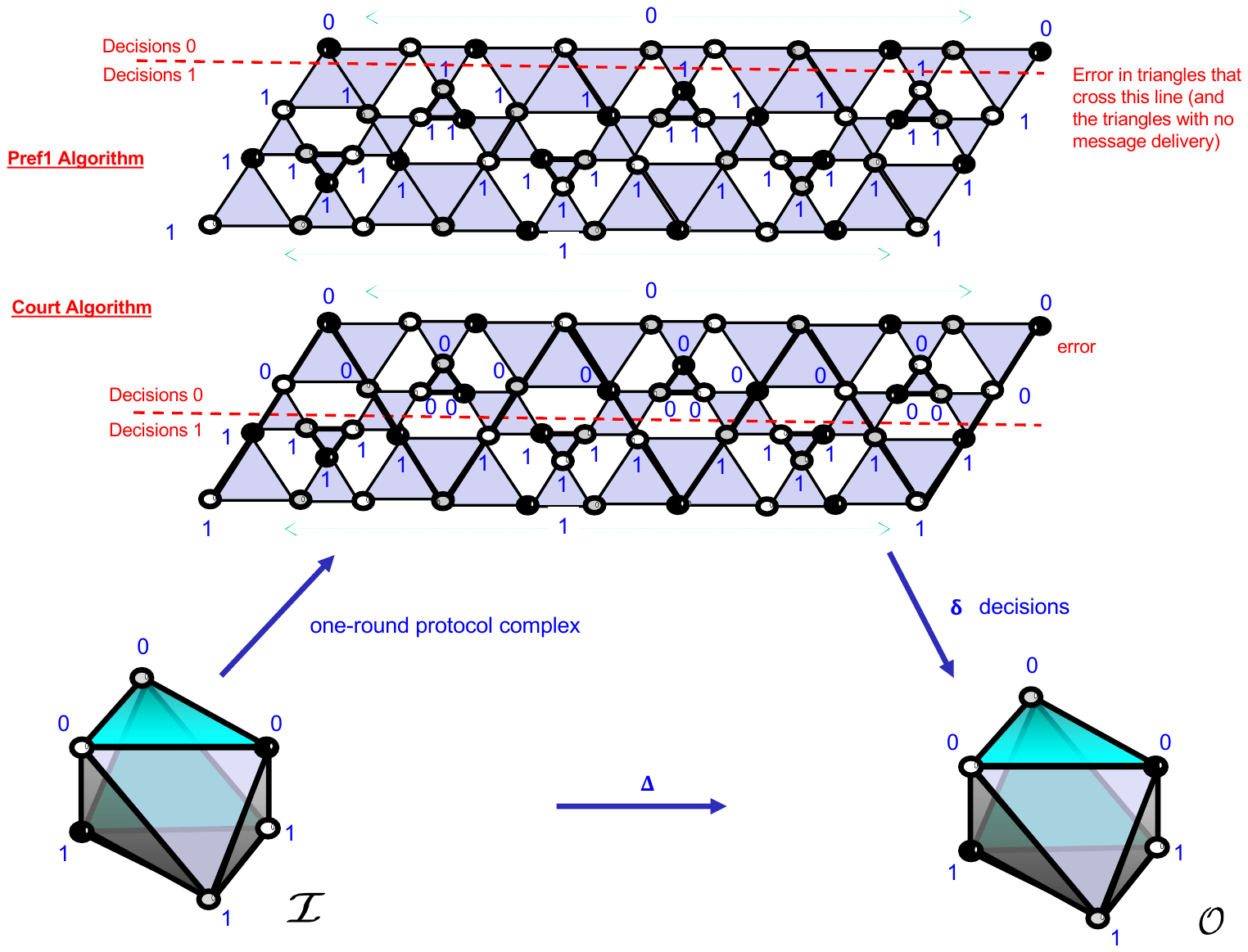}}
\caption{Bottom: For $n=3$, the  input complex $\mathcal{I}$  and output complex $\mathcal{O}$   of all binary values. 
Top: The $1$-round protocol complex $\mathcal{P}$, including
all triangles except those where no messages are delivered. 
In each of the two top figures, the left side and right side are identified in the protocol complex. 
The top subcomplex (resp., bottom subcomplex) corresponding to the input configuration $\mathbf{x}=(000)$ (resp., $\mathbf{x}=(111)$) is not represented. 
The set of triangles for which $\pref1$ and $\cort$ fail to solve consensus induce a cut of the protocol complex, shown as red dotted line. 
Each triangle is one possible execution, and is therefore associated with a probability. The question is: what is the cut minimizing the probability of error over all possible inputs? Equivalently, what is the chromatic simplicial map $\Delta$ satisfying validity, with the smallest (weighted by probability) number of bicolored output simplexes? }
    \label{fig:optAlgsComplex-mainText}
\end{figure}

A closer look at the one-round protocol complex  reveals that, inside the complex for fixed input values,
two triangles intersect in a single vertex: when a process receives messages from two processes, it knows
the full input vector; its only uncertainty is if  its own broadcast was received or not.
This  suggests replacing the protocol complex, with its dual, called  the \emph{Kripke graph} (each vertex of the Kripke graph corresponds to a triangle of the protocol complex).
Thus moving from 
the ``vertex problem,'' which asks for the best assignment of output 0/1-values to the vertices of the protocol complex~$\mathcal{P}$, with an equivalent ``edge problem,'' in which 0/1-values are assigned to the edges of the Kripke graph.
We stress that moving to the Kripke graph can be done in any model of computation~\cite{GOUBAULT2021104597}, but it is in the broadcast model that  the Kripke graph has these nice properties. 

However, the Kripke graph $\mathcal{K}$  of the stochastic broadcast model is not easy to analyze directly, in particular because
it represents all one-round executions, over \emph{all} input value assignments
(so it has parallel edges, see Section~\ref{sec:cubes}).
A main insight is that we can work with what we call the \emph{reduced Kripke graph}, denoted by~$\mathcal{Q}$. 
This graph has two highly desirable properties: 
First, the error probability of an algorithm corresponds to the weight of the \emph{vertex cut},
    induced by the vertices incident on edges of both colors.
 Second,  $\mathcal{Q}$ has a hypercube structure, which enables a direct geometric embedding  of $\mathcal{Q}$ in the $n$-dimensional Euclidean space, and facilitates its analysis. See Fig.~\ref{fig:reduced-Kripke-graph-mainText}. 
 
\begin{figure}[htb]
\centering
\includegraphics[width=5cm,height=4.5cm]{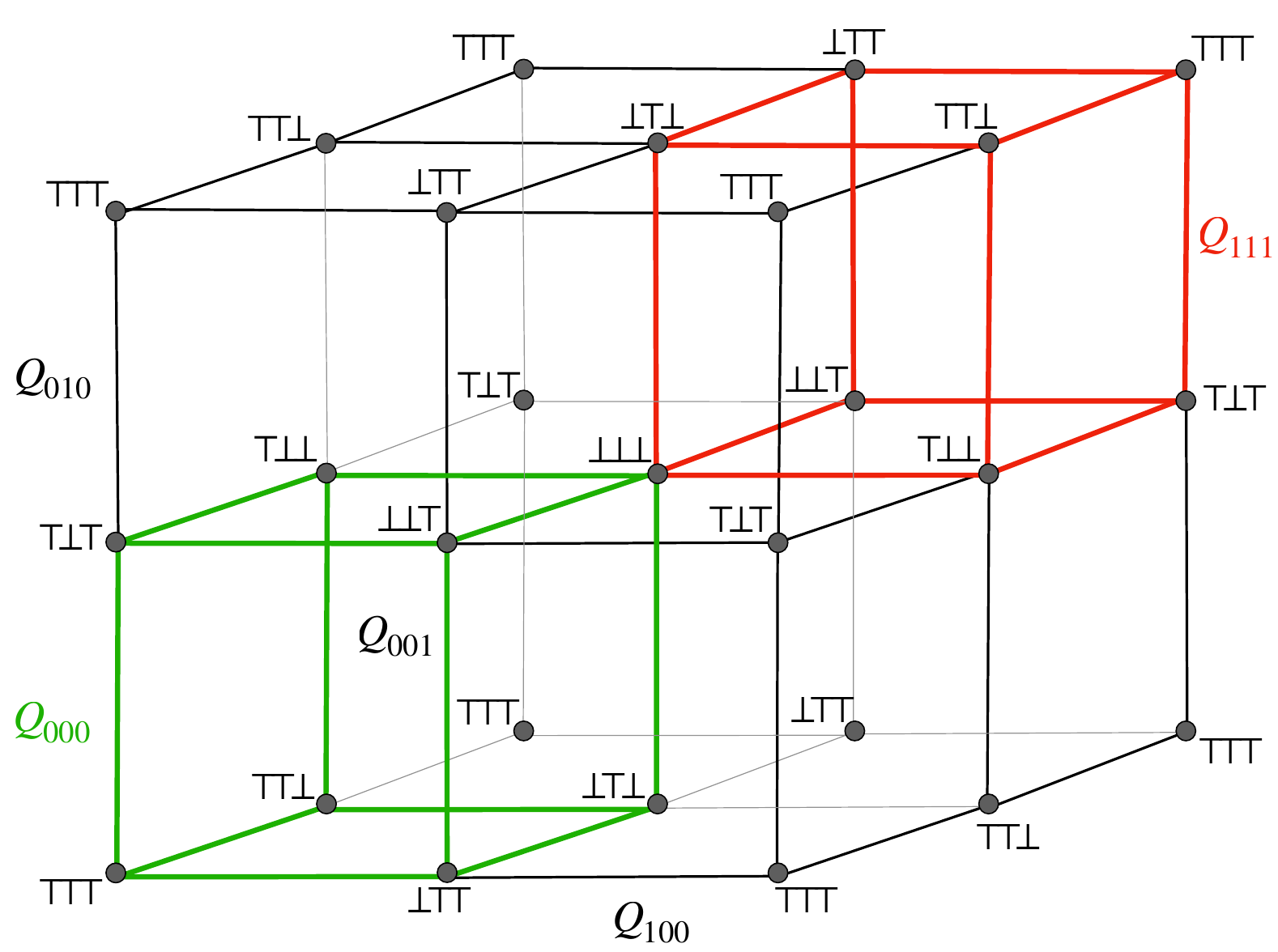}
\caption{The reduced Kripke graph $\mathcal{Q}$ for three processes. It is composed of 8 cubes $Q_{x,y,z}$, each corresponding to a different input $(x,y,z)\in\{0,1\}^3$, whose vertices are failure patterns in $\{\bot,\top\}^3$. The probability of a failure pattern is $p^kq^{3-k}$ where $k$ denotes the number of $\top$ symbols in the  pattern. The horizontal, vertical, and diagonal axes are relative to processes~1, 2, and~3, respectively.}
\label{fig:reduced-Kripke-graph-mainText}
\end{figure}

%%%%%%%%%%%%%%%%%%%%%%%%%%%%%%%%%
\subsection{Related Work}
%%%%%%%%%%%%%%%%%%%%%%%%%%%%%%%%%

Many papers have considered synchronous models where links can fail probabilistically, especially to solve broadcast and gossip, but not to solve consensus, to the best of our knowledge.
For an early survey, see Pelc~\cite{pelcSurvey1996}. The  case where a process may fail to send in a round probabilistically has also been considered, e.g., Pelc and Peleg~\cite{PelcPeleg2005}, they assume the same model we consider but over a communication graph that is not necessarily complete. However, they did not study consensus; instead,   they   estimate the number of rounds needed for broadcasting with high probability.

Many other 
stochastic versions of synchronous models have been studied,  where in each round, a directed graph is chosen  at random among graphs from a given set $S$, such as
 directed and undirected Erdős–Rényi graphs,  directed rooted trees, or directed graphs with at least $k$ edges. Some surveys and recent papers are~\cite{SmoothedAnalDynNet2024,elhayekDisc2024,survey1988,KuhnO11,MichailSpirakis2018}.

Previous papers usually prove asymptotic results,  essentially implying  that  processes can, with high probability, learn each other's inputs and solve consensus in a small  number of rounds in the stochastic broadcast model. However, this type of results  are not applicable for a small, constant number of processes $n$. For instance,
Alon et al.~\cite{AlonBEGH19} assume  binary symmetric channels (BSC), i.e., each message is a bit that may be flipped with some constant probability $\varepsilon>0$. 
 Roughly, they show that for $n\to\infty$, any computation over failure-free channels can be emulated, with high probability, over BSC
channels with a constant multiplicative overhead.  
Another example is El-Hayek, Henzinger and Schmid~\cite{elhayekDisc2024}, who consider stochastic adversaries and show that consensus can be solved with high probability in a logarithmic number of rounds. 

A notable exception to asymptotic analysis papers is~\cite{FraigniaudPSR25}, which initiates the study of the performance of stochastic dynamic networks for specific values of $n$, with a thorough description of the case of $n=2$, as discussed 
above.

The {consensus}  problem has been thoroughly studied before using randomized algorithms in two settings: with probabilistic termination and under a fixed number of rounds.
There is much work on randomized algorithms for consensus with no agreement error (in a non-stochastic setting), but requiring termination only with high probability  e.g.~\cite{AttiyaW04} and~\cite[Chapter 24]{aspnes2025notestheorydistributedsystems}.
With fixed termination, allowing agreement error as we do, the case of  $n=2$ synchronous reliable processes, when  messages can
be lost, is called the \emph{coordinated attack}~\cite{Gray78} (or the two generals problem), and more generally,
the  \emph{randomized coordinated attack}  for an arbitrary number of processes~\cite{VARGHESE199657} and~\cite[Chapter 8]{aspnes2025notestheorydistributedsystems}. These papers have assumed  randomized algorithms in an adversarial setting, while we assume  deterministic algorithms in a stochastic setting. 

The protocol complex representation of a full-information algorithm has been very useful~\cite{HerlihyKR13}. The vertices are local states, and the facets are global states. Its dual is a Kripke graph, where vertices, corresponding to facets, are global states. 
This duality has been formalized by Goubault, Ledent and Rajsbaum~\cite{GOUBAULT2021104597}, and given an epistemic logic semantics.

We use this duality in the stochastic broadcast model, observing that the Kripke graph after one round, for a fixed input vector, is a hypercube.  
 This is exactly the Kripke graph for the  initial states of the Muddy Children puzzle, which has been known
 since Halpern and Vardi~\cite{halpernVardi91}, described in detail in the subsequent book~\cite[Figure 2.3]{knowledgeBook}.
Interestingly, in our model the processes do know their  inputs (in the Muddy Children puzzle players know all inputs except their own). Thus, in our case, the initial Kripke graph is \emph{not} a hypercube. 

The Muddy Children complex is natural, and it is
not a surprise to find that it has appeared in other papers studying unrelated problems: it has also appeared in the local generation of certain languages~\cite{hoyrup2025localgenerationlanguages}.

Another research line using simplicial complexes and randomization is by Fraigniaud, Gelles and Lotker~\cite{FraigniaudGL24}. They also consider synchronous, failure-free processes, but assume the processes are anonymous and the system has access to $k$ independent sources of randomness. They study eventual solutions to symmetry-breaking tasks.

%%%%%%%%%%%%%%%%%%%%%%%%%%%%%%%%%
\subparagraph{Organization.}
%%%%%%%%%%%%%%%%%%%%%%%%%%%%%%%%%

The model is presented in~\sectionref{sec:model}.
The two main  algorithms $\pref1$ and $\court$ are presented and analyzed in \sectionref{sec:algs3}.
  In \sectionref{sec:cubes} we present the reduced Kripke
graph, and use it to prove the main lower bound.
In \sectionref{sec:gen} we consider the general case of $n>3$ and multiple rounds. 
We conclude in \sectionref{sec:conc} with a summary and open problems. 
In addition, \appref{sec:Protocol-Complex} contains
a detailed description of the tight connections between the topological approach of distributed computing via the analysis of the protocol complex, and the epistemic approach of distributed computing via the analysis of the Kripke graph in the context of the stochastic broadcast model. 

%%%%%%%%%%%%%%%%%%%%%%%%%%%%%%%%%
\section{Model and Preliminaries}
\label{sec:model}
%%%%%%%%%%%%%%%%%%%%%%%%%%%%%%%%%

\subsection{The Stochastic Broadcast Model}

Consider a set of $n\geq 1$ processes labeled from 1 to~$n$. 
We denote $[n]=\Set{1,2,\ldots,n}$ for $n\in\Nat$.
For each $i\in [n]$, there is a broadcast channel on which process  $i$ sends messages, and all other processes are receivers.

The  \emph{stochastic broadcast} model is defined by a number of processes $n\geq 2$, and a probability parameter $p\in[0,1]$.
    Communication proceeds as a sequence of synchronous rounds. At each round $r\geq 1$, each of the $n$ processes can send a message of arbitrary size 
    on its broadcast channel.
    Each message broadcast by a process in any round may be either \emph{successful}, with probability $p\in[0,1]$, or unsuccessful, with probability $q=1-p$.
    A successful message is delivered to all non-sender processes, whereas an unsuccessful message is dropped and no process receives it. 
    A sender has no indication
    whether its broadcast is successful or not (otherwise solving consensus is trivial).

Note that, as the size of the message sent by any process at any given round is unlimited, different data can be communicated to different processes, whenever necessary, by marking which part of the message is for which process.

Note also that even though a process
 does not immediately know whether its message was received or not, 
it may however
learn it later, from other processes.
On the other hand, a process $i$ receiving of a message from a process $j\neq i$ knows that all the other processes have received the message from~$j$. Conversely, if a process$i$ does not receive any message from process $j$ at round~$r$, then $i$ knows that none of the processes have received a message from~$j$. 

This paper focuses on \emph{deterministic} algorithms. 
The outputs of a deterministic algorithm in the stochastic broadcast model may be random due to the probabilistic nature of the message delivery. 
We  consider \emph{one-shot algorithms}, i.e., algorithms that execute  a prescribed  number of communication rounds, and then each process outputs a value.
Specifically, we look for consensus algorithms with optimal probability of error (validity is always enforced;   only  agreement can be incorrect), which terminate in a given number of rounds~$r$. 

%%%%%%%%%%%%%%%%%%%%%%%%%%%%%%%%%
\subsection{1-Round Executions of 3 processes}
\label{sec:1-round-algo-operational}
%%%%%%%%%%%%%%%%%%%%%%%%%%%%%%%%%

Given an assignment of inputs to the processes, we have a well-defined
probability space over executions.
For example, the single-round
execution in which
all messages are dropped has probability $q^n$, and the probability
that exactly one process succeeds in sending its messages  in a single-round execution
is $npq^{n-1}$. In general, the probability of an execution $v$, where $k$ given messages are delivered and the other $n-k$ messages are lost,
is $\Pr[v]=p^kq^{n-k}$. 
In each $1$-round consensus algorithm, processes decide binary values. The probability of error
is equal to the number of executions where there is disagreement, weighted by their corresponding probabilities.

\subparagraph{Notation.}

Let us fix a $1$-round execution of a protocol.
\begin{itemize}
        \item $\In_i\in\Set{0,1}$ and $\Out_i\in\Set{0,1}$ respectively denote the input and the output of process~$i$.
        \item $R$ denotes the set of messages successfully delivered.
        \item $R_0\subseteq R$ denotes the set of $0$-message successfully delivered,
        and we set $R_1= R\smallsetminus R_0$.
        \item $count_i(0)$ denotes the number of  inputs $0$ that process $i$ \emph{knows about} (recall that process $i$ does not know whether $R$ includes its own transmission), i.e.,
        $$count_i(0)=\begin{cases}
            \left|R_0\smallsetminus\Set i\right|\,+\,1 &\text{if $\In_i=0$,}\\
            |R_0| &\text{if $\In_i=1$.}
        \end{cases}
        $$
       Similarly, 
        $$count_i(1)=\begin{cases}
            |R_1| &\text{if $\In_i=0$,}\\
            \left|R_1\smallsetminus\Set i\right|\,+\,1 &\text{if $\In_i=1$.}
        \end{cases}
        $$
\end{itemize}

\subparagraph{Example: Algorithm \textsf{majority}.}

Consider, for instance, a typical 1-round consensus algorithm, where all processes broadcast their input values, and decide on the majority of the input values 
they know about, and decide $0$ in case of a tie. 
We formalize the specification in Algorithm~\ref{alg:majority}.

\begin{algorithm}[htb]
\caption{\textsf{majority} --- Code of process $i\in [n]$}
\label{alg:majority}
\begin{algorithmic}[1]
\State broadcast $\In_i$, and receive messages
\State $\displaystyle
\Out_i=\begin{cases}
0&\text{if $count_i(0)\ge count_i(1)$}\\
1&\text{otherwise}
\end{cases}
$
\end{algorithmic}
\end{algorithm}

We stress that we look for algorithms with the smallest probability of error on the worst-case input vector.
By enumerating all executions, we can compute the error probability.
These can be visualized in the $1$-round protocol complex depicted in \figref{fig:simplex-flat-proba-mainText}, where the facets (triangles) are annotated with  delivery patterns. 
It is straightforward to show that the worst-case 
error probability is attained for input vector $011$,
the error probability of \textsf{majority} is $2p^2q+2pq^2+q^3$.

%%%%%%%%%%%%%%%%%%%%%%%%%%%%%%%%%%%%%%%%
\section{Algorithms for   Consensus: $\court$ and $\pref$ }
\label{sec:algs3}
%%%%%%%%%%%%%%%%%%%%%%%%%%%%%%%%%%%%%%%%

In this section we present  one-round consensus algorithms in the stochastic broadcast model,
analyzing their worst-case error probability.
Their behavior is plotted in Figure~\ref{fig:plot-cort2To6}.

\subsection{Algorithms $\pref0$ and $\pref1$}

The two algorithms $\pref$ ($\pref0$ and $\pref1$) are biased toward a particular value:
$\pref0$, for example, says that every process outputs $0$ unless the process hears only $1$ values:

\begin{algorithm}[h]
\caption{$\pref0$ --- Code of process $i\in [n]$}
\label{alg:pref0}
\begin{algorithmic}[1]
\State broadcast $\In_i$, and receive messages
\State \label{line2}$\displaystyle
\Out_i=\begin{cases}
0&\text{if $count_i(0)>0$}\\
1&\text{otherwise}
\end{cases}
$
\end{algorithmic}
\end{algorithm}

Algorithm $\pref1$ is  the same as \algoref{alg:pref0}, except that 
line~\ref{line2} is replaced by the following:
\[
\text{\hspace*{-3mm}\small \ref{line2}}\!\!:~\Out_i=\begin{cases}
1&\text{if $count_i(1)>0$}\\
0 &\text{otherwise}
\end{cases}
\]

\begin{theorem}
\label{thm-pref}
The worst-case error probability of $n$-player, one-round algorithms $\pref0$ and $\pref1$ is $q$.
\end{theorem}

\begin{proof}
    Let us consider, without loss of generality, $\pref0$, and 
    suppose that the input contains $d$ zero values and $n-d$ one's, for some $0<d<n$.
    By the specification of the algorithm, all processes with $0$ input decide $0$. Processes
    with input $1$ decide $0$ if $R_0$ is non-empty, i.e., if
    any of the $0$ transmissions is  successful. It follows
    that with this input, Algorithm $\pref0$ fails with probability $q^d$, which
    is maximized when $d=1$.
\end{proof}

%%%%%%%%%%%%%%%%%%%%%%%%%%%%%%%%%
\subsection{Algorithm $\court$}
\label{ssec:court}
%%%%%%%%%%%%%%%%%%%%%%%%%%%%%%%%%

Algorithm $\court$ is defined as follows. Every process counts the number of $0$ values, and
the number of $1$ values it knows about, i.e., its own input value,
and the values that were successfully transmitted. See Algorithm~\ref{alg:corteous} for a formal description of $\court$ where $\neg x=1-x$ for $x=0,1$. 

\begin{algorithm}[h]
\caption{$\court$ --- Code of process $i\in [n]$}
\label{alg:corteous}
\begin{algorithmic}[1]
\State broadcast $\In_i$, and receive messages
\State $\displaystyle
\Out_i=\begin{cases}
    0 &\text{if $count_i(0)>count_i(1)$}\\
    1 &\text{if $count_i(1)>count_i(0)$}\\
    \neg\In_i&\text{if $count_i(0)=count_i(1)$}
\end{cases}
$
\end{algorithmic}
\end{algorithm}

Intuitively, $\court$ follows the majority, but in case of a tie, it outputs the negation of its own input value
(hence the name). It turns out that this (possibly counterintuitive) tie-breaking rule
has very pleasing consequences. Let us analyze the situations in which Algorithm $\court$
fails. Since the above formulation of $\court$
extends to any number $n$ of processes, we present the 
general analysis.

\begin{lemma}
  \label{lem:maj}
  If $|R_0|\ne|R_1|$ i.e., if one of the values has a majority in the set of successfully delivered
  messages,
  then  $\court$ reaches
  agreement on that value.
\end{lemma}

\begin{proof}
 Assume that $|R_0|>|R_1|$. 
We proceed by considering process behavior according to whether the
transmission of the process was delivered or dropped. 
\begin{description}
\item[Case $i\notin R$:] Such a process $i$ received all
successful messages and hence received a
  majority of $0$ messages. Therefore, $count_i(0)\ge count_i(1)$ regardless of $\In_i$. If $count_i(0)>count_i(1)$ then
  $\Out_i=0$; and if $count_i(0)=count_i(1)$ then it must be the case that $\In_i=1$, and therefore $\Out_i=\neg\In_i=0$ as well.
\item[Case $i\in R$:] 
for any such  process $i$ we have $count_i(0)=|R_0|$ and $count_i(1)=|R_1|$.
By assumption that $|R_0|>|R_1|$ we therefore have that $count_i(0)>count_i(1)$ and hence $\Out_i=0$.
\end{description}
Thus, if $|R_0|>|R_1|$, all processes output $0$. 
The case of $|R_0|<|R_1|$ is analogous: it results in all processes outputting $1$.
\end{proof}

\begin{corollary}
  If the number of delivered messages is odd, then $\court$ results in
  agreement. 
\end{corollary}

\begin{lemma}
  \label{thm:court-prob}
  Consider an execution of $\court$ with input assignment that has $n_0$
  $0$-values and $n_1=n-n_0$ $1$-values, and assume w.l.o.g.\ that
  $0<n_0\le n_1<n$. Then
  $$
  \Pr[\mbox{$\court$ fails}]=
  \sum_{i=0}^{n_0}{n_0\choose i}{n_1\choose i}p^{2i}q^{n-2i}~.
  $$
\end{lemma}

\begin{proof}
By \lemmaref{lem:maj}, $\court$ does not fail if the values that are
successfully transmitted are not split equally. 
Suppose now that in the set of successfully transmitted messages,
the number of 0 votes  equals the number of 1 votes.
Then every successful transmitter $j$
has $c_j(0)=c_j(1)$, hence it sets $\Out_j=\neg\In_j$. In addition,
every process whose transmission failed
decides its input value. It follows that  $n_0$ 
processes decide $0$ and $n_1$ processes decide $1$. 
By assumption,  $n_0>0$ and $n_1>0$, therefore
we
conclude that $\court$ fails when the number of 0 votes  equals the number of 1 votes. To calculate 
the probability of failure, note the probability
that exactly $i$ $0$-broadcasts
and $i$ $1$-broadcasts are successful is
${n_0\choose i}{n_1\choose i}p^{2i}q^{n-2i}$,
and the theorem follows.
\end{proof}

\begin{theorem}
\label{cor:equal}
    The worst-case error probability of $n$-player,   algorithm  $\court$  is %at most
    $$
    \sum_{i=0}^{\floor{n/2}}{\floor{n/2}\choose i}{\ceil{n/2}\choose i}p^{2i}q^{n-2i}~.
  $$
\end{theorem}

\begin{proof}
We show that the worst-case input is when the input values are divided
as equally as possible: $\floor{n/2}$ start with $0$ and $\ceil{n/2}$
start with $1$ (or vice versa). To see this, we use the fact that for
any $n\in\Nat$ and integer $0\le k\le n$,
\begin{equation}
  \label{eq:CS}
 k(n-k)\le\floor{n\over2}\cdot\ceil{n\over2}~. 
\end{equation}
Now, consider any nontrivial input with $0<k<n$ 0's and $n-k$
1's. Assume without loss of generality that $k\le n-k$, i.e.,
$k\le\floor{n/2}$.
Then the probability that $\court$ fails on that input is exactly
\begin{align*}
  \sum_{i=0}^{k}
  &\binom{k}{i}\binom{n-k}{i}p^{2i}q^{n-2i}
  \\
  & =\sum_{i=0}^k\frac{k(k-1)\cdots(k-i+1)}{i!}\cdot\frac{(n-k)(n-k-1)\cdots(n-k-i+1)}{i!}p^{2i}q^{n-2i}
  \\
  &=\sum_{i=0}^k\frac{k(n-k)\cdot(k-1)(n-k-1)\cdots(k-i+1)(n-k-i+1)}
    {(i!)^2}p^{2i}q^{n-2i}
  \\
&\le\sum_{i=0}^k\frac{\floor{n\over2}\ceil{n\over2}\cdot
    (\floor{n\over2}-1)(\ceil{n\over2}-1)\cdots(\floor{n\over2}-i+1)(\ceil{n\over2}-i+1)}
    {(i!)^2}p^{2i}q^{n-2i}
  &\text{by \Eqr{eq:CS}}
  \\
  & \le \sum_{i=0}^{\floor{n/2}}
    \binom{\floor{n/2}}{i}\binom{\ceil{n/2}}{i}p^{2i}q^{n-2i}~,
\end{align*}
as claimed.
\end{proof}

\begin{figure}[htb]
\centering
\includegraphics[scale=.55]{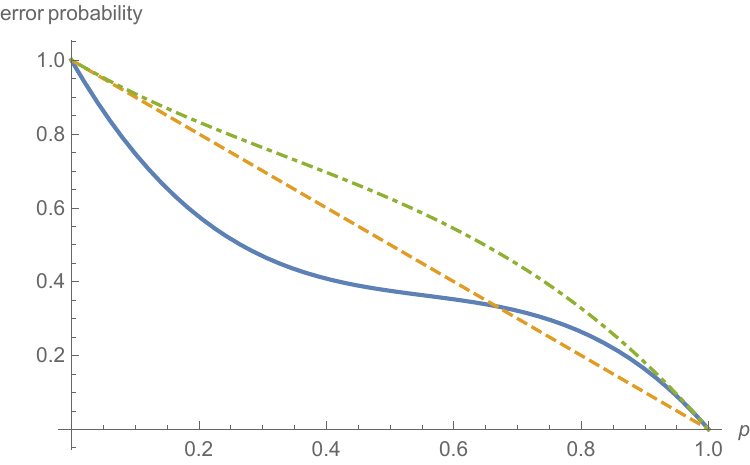}
\caption{The error probability of the  $3$-player algorithms: $\court$
(thick blue line), $\pref1$ (dashed orange line) and \textsf{majority} (dashed-dotted green line).}
\label{fig:maj}
\end{figure} 

We note that for $p=1/2$, the error probability formula  simplifies to 
$
\binom{n}{\floor{n/2}}2^{-n}\in \Theta(1/\sqrt n).
$
\figref{fig:maj} shows the worst-case error probabilities of
\textsf{majority}, $\pref1$ and $\court$ for $n=3$. 

\begin{figure}[htb]
\centering
\includegraphics[scale=.4]{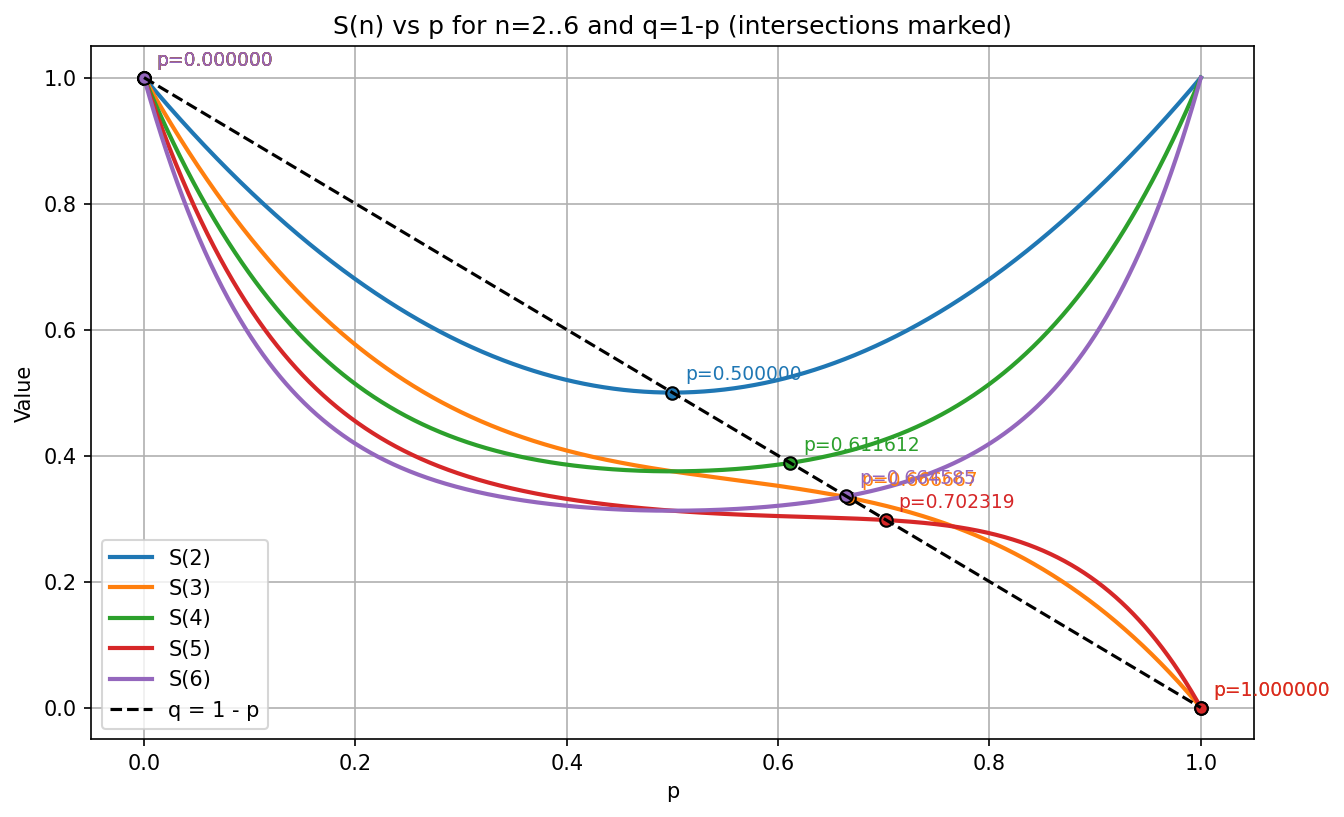}
\caption{The behavior of $\cort$ for small values of $n$. }
\label{fig:plot-cort2To6}
\end{figure}

See Figure~\ref{fig:plot-cort2To6} for graphs of the error probability
of $\cort$. Note the different behaviors for odd and even $n$ values.

%%%%%%%%%%%%%%%%%%%%%%%%%%%%%%%%%
\section{Lower Bound for 1-Round Consensus in 3-Process Systems}
\label{sec:cubes}
%%%%%%%%%%%%%%%%%%%%%%%%%%%%%%%%%

This section is dedicated to proving a structural theorem establishing a one-to-one correspondence between 1-round consensus algorithms in the stochastic broadcast model, and 2-edge-colorings of a certain graph. 
This is the basis for our analysis of the lower bound on the error probability of 1-round consensus algorithms. 

\subsection{The Kripke Graph $\mathcal{K}$} 
\label{sec:kripkeGraph}

The vertices of the Kripke graph~\cite{knowledgeBook,kripke1959}  $\mathcal{K}$ after one round are all the possible configurations of the system after one round. These vertices can thus be defined as all pairs $(\mathbf{x},\phi)$ where $\mathbf{x}\in\{0,1\}^n$ is an $n$-dimensional vector specifying the inputs assigned to the processes, and $\phi\in\{\bot,\top\}^n$ is an $n$-dimensional  vector specifying the message deliveries. There is an edge between two vertices $(\mathbf{x},\phi)$ and $(\mathbf{x}',\phi')$ if there exists a process that cannot distinguish these two configurations. Such an edge is then labeled by the set of processes that cannot distinguish the two configurations. 

Given $\mathbf{x}\in\{0,1\}^n$, the subgraph $\mathcal{K}(\mathbf{x})$ of the Kripke graph $\mathcal{K}$ induced by the $2^n$ vertices $(\mathbf{x},\phi)$ with $\phi\in\{\bot,\top\}^n$ is isomorphic to the $n$-dimensional hypercube with vertex set $\{\bot,\top\}^n$, and edge set consisting of all vertex pairs
$\{\phi,\phi'\}$ that differ in exactly one entry. Moreover, if $\phi$ and $\phi'$ differ in position~$i$, then the edge $\{(\mathbf{x},\phi),(\mathbf{x},\phi')\}$ of $\mathcal{K}(\mathbf{x})$ is labeled~$i$. So, in $\mathcal{K}(\mathbf{x})$, every edge is labeled by exactly one process name $i\in [n]$, which can also be interpreted as the dimension of that hypercube edge. 

It follows that the Kripke graph $\mathcal{K}$ can be decomposed into $2^n$ vertex-disjoint subgraphs $\mathcal{K}(\mathbf{x})$, each one isomorphic to the $n$-dimensional hypercube. The Kripke graph also contains edges between vertices of different hypercubes $\mathcal{K}(\mathbf{x})$ and $\mathcal{K}(\mathbf{x}')$ for $\mathbf{x}\neq \mathbf{x}'$. 
However,  
we will show that these edges can be ignored, and the subcubes can be ``glued'' together, without information loss (w.r.t.\ the Kripke graph), as far as the design of consensus algorithms with low error probability is concerned. 

%%%%%%%%%%%%%%%%%%%%%%%%%%%%%%%%%
\subsection{The Reduced Kripke Graph $\mathcal{Q}$} 
\label{sec:ReducedKripkeGraph}
%%%%%%%%%%%%%%%%%%%%%%%%%%%%%%%%%

For the purpose of defining the reduced Kripke graph, denoted by $\mathcal{Q}$, it is convenient to embed the vertices of that graph in~$\mathbb{R}^n$. Specifically, the vertices of~$\mathcal{Q}$ are $\{-1,0,1\}^n$. Note that $\mathcal{K}$ has more than $3^k$ vertices, meaning that more than one vertex may be embedded at the same coordinates in~$\mathbb{R}^n$, i.e., to the same vertex in~$\mathcal{Q}$. 

\subsubsection{Construction of the Reduced Kripke Graph}

Let us define, for every bit $b\in\{0,1\}$, and every delivery symbol $\delta\in\{\bot,\top\}$, the function 
\[
g(b,\delta)=\begin{cases}
    0 &\text{if $\delta=\bot$}\\
    -1 &\text{if $\delta=\top$ and $b=0$}\\ 
     1 &\text{if $\delta=\top$ and $b=1$}. 
  \end{cases}
\]
Every vertex $(\mathbf{x},\phi)\in\{0,1\}^n\times \{\bot,\top\}^n$ of $\mathcal{K}$ with $\mathbf{x}=(x_1,\dots,x_n)$ and $\phi=(\phi_1,\dots,\phi_n)$ is mapped to the point 
\[
g(\mathbf{x},\phi)\DEF (g(x_1,\phi_1),\dots,g(x_n,\phi_n))\in\{-1,0,1\}^n.
\]
In particular, for every $\mathbf{x}=(x_1,\dots,x_n)\in\{0,1\}^n$, the subgraph $\mathcal{K}(\mathbf{x})$ of $\mathcal{K}$ is thus mapped to the geometric hypercube 
\[
Q_\mathbf{x}=[0,g(x_1,\top)]\times \dots \times [0,g(x_n,\top)].
\]
Note that, for every vertex $v$ of~$\mathcal{Q}$, i.e., for every position $\mathbf{v}=(v_1,\dots,v_n) \in\Set{-1,0,1}^n$,  all vertices of $\mathcal{K}$ that are mapped to $\mathbf{v}$ have the same delivery pattern $\varphi(\mathbf{v})$ defined as 
\begin{equation}
\label{eq:delivery-pattern-associated-to-vertices}
\varphi(\mathbf{v})\DEF (\varphi(v_1),\dots,\varphi(v_n))~,
\mbox{ where, for $i\in [n]$,\hspace*{2mm} }
\varphi(v_i)=\begin{cases}
    \top &\text{if $|v_i|=1$}\\
    \bot &\text{if $v_i=0$}~.
\end{cases}
\end{equation}
The weight associated with a point $\mathbf{v}\in\Set{-1,0,1}$ is $p^{\|\mathbf{v}\|}q^{n-\|\mathbf{v}\|}$, where $p$ is the model probability parameter, and $\|\mathbf{v}\|=\sum_{i=1}^n|v_i|$, i.e., $\|\mathbf{v}\|$ is
the number of $\top$ symbols in $\varphi(\mathbf{v})$. (The weight
of a point in $\mathcal{Q}$ is the probability that its associated delivery 
pattern occurs.)
The resulting graph  $\mathcal{Q}$ in $\mathbb{R}^n$ is called the \emph{reduced Kripke graph}.
See Figure~\ref{fig:reduced-Kripke-graph-mainText}. 

\subsubsection{Properties of the Reduced Kripke Graph}

It follows from the embedding that the vertices of $\mathcal{Q}$ are the $3^n$ points  $v\in\mathbb{R}^n$ at coordinates $\{-1,0,1\}^n$, and there is an edge between two vertices $\mathbf{v}=(v_1,\dots,v_n)$ and $\mathbf{v}'=(v'_1,\dots,v'_n)$ of $\mathcal{Q}$ if and only if $\mathbf{v}$ and $\mathbf{v}'$ differ in a single entry~$i$ by $1$, i.e., for some index $i\in[n]$,  $|v_i-v'_i|=1$, and for all $j\ne i$,
$v_j=v'_j$. The index $i$ is called the \emph{dimension} of the edge $\{\mathbf{v},\mathbf{v}'\}$ of~$\mathcal{Q}$. 

Before analyzing the semantics of $\mathcal{Q}$, let us consider 
its geometry briefly. Let us denote by $\mathcal{Q}(\mathbf{x})$ the embedding of $\mathcal{K}(\mathbf{x})$ in $\mathbb{R}^n$. Note that while $\mathcal{K}(\mathbf{x})$ and $\mathcal{K}(\mathbf{x}')$ are vertex-disjoint for $\mathbf{x}\neq \mathbf{x}'$, the corresponding cubes $\mathcal{Q}(\mathbf{x})$ and $\mathcal{Q}(\mathbf{x}')$ share vertices. For instance, the
origin of $\mathbb{R}^n$, which is mapped by $\varphi$ to the pattern~$\bot^n$, is common to all the hypercubes $\mathcal{Q}(\mathbf{x})$ --- it is actually the only vertex common to all $\mathcal{Q}(\mathbf{x})$'s. Conversely, each ``corner vertex,''
i.e., each vertex $\mathbf{v}\in\Set{-1,1}^n$ of $\mathcal{Q}$,  is mapped to the
pattern~$\top^n$, and belongs to a single $\mathcal{Q}(\mathbf{x})$. 

We now show that no information about the different executions (inputs and delivery patterns)
has been lost  by considering the reduced Kripke graph $\mathcal{Q}$ instead of the ``complete'' Kripke graph $\mathcal{K}$, as far as algorithm design is concerned. 

\begin{lemma}\label{lem:reduced-kripke-graph-preserve-information}
  Let $e=\{\mathbf{v},\mathbf{v}'\}$ be an edge of the reduced Kripke graph $\mathcal{Q}$, of dimension~$i\in [n]$, and let $\mathbf{x}\in\{0,1\}^n$ such that $e \in E(\mathcal{Q}(\mathbf{x}))$.
  \begin{itemize}
    \item For the input $\mathbf{x}$, process $i$ has the same local view in both executions $\varphi(\mathbf{v})$ and $\varphi(\mathbf{v}')$. 

    \item If there exists $\mathbf{x}'\neq \mathbf{x}$ such that $e\in E(\mathcal{Q}(\mathbf{x}))\cap E(\mathcal{Q}(\mathbf{x}'))$ then process $i$ has the same view in all  executions that start with input assignment $\mathbf{x}$ or $\mathbf{x}'$, and have delivery patterns $\varphi(\mathbf{v})$ or $\varphi(\mathbf{v}')$.
    
    \item  Let $e'=\{\mathbf{w},\mathbf{w}'\}$ be an edge of the reduced Kripke graph $\mathcal{Q}$, distinct from~$e$, but of same dimension~$i\in [n]$. The local view of process $i$ in $e$ is different from the local view of process $i$ in $e'$.
  \end{itemize}
\end{lemma}

\begin{proof}
For establishing the first item, it is sufficient to observe that, under a fixed input,
the only difference between the executions represented by $\varphi(\mathbf{v})$ and $\varphi(\mathbf{v}')$ is that in one the message broadcasted by process $i$ was successfully
delivered, while in the other the message was dropped. These events are
indistinguishable to process~$i$.

Thanks to the first item, for establishing the second item, it is sufficient to
show that, for the pattern $\varphi(\mathbf{v})$, the view of process $i$
is the same under input assignments $\mathbf{x}$ and $\mathbf{x}'$.
By construction, $\mathbf{v}$ and $\mathbf{v}'$ differ only in
coordinate~$i$, in which one of them has coordinate~$0$, and the other has
coordinate either $1$ or~$-1$.
Moreover, also by construction, since the edge $e$ has dimension~$i$, and belongs both to $\mathcal{Q}(\mathbf{x})$ and $\mathcal{Q}(\mathbf{x}')$, it must be the case that $x_i=x'_i$,
i.e., the local input of $i$ is the same in $\mathbf{x}=(x_1,\dots,x_n)$ and $\mathbf{x}'=(x'_1,\dots,x'_n)$. 
Let $\bar{Z} \subseteq[n]\smallsetminus\Set{i}$ denote the set of non-zero coordinates of~$\mathbf{v}$, other than~$i$.
Since for all $j\ne i$ we have $v_j=v'_j$, it follows that
$x_j=x'_j$ for all $j\in \bar{Z}$. Moreover,
the messages that are successfully delivered in
$\varphi(\mathbf{v})$ and $\varphi(\mathbf{v}')$ (with the exception of process~$i$)
are sent exactly by the  processes in $\bar{Z} $. It follows that process $i$ has
the same view in the execution with input $\mathbf{x}$ and delivery pattern
$\varphi(\mathbf{v})$ as in the execution with input $\mathbf{x}'$ and delivery pattern
$\varphi(\mathbf{v})$, as required.

For the third item, note that if $e$ and $e'$ are in the same subgraph $\mathcal{Q}(\mathbf{x})$, then the delivery patterns
$\varphi(\mathbf{v})$ and $\varphi(\mathbf{w})$ differ in the delivery status of some messages from processes other than~$i$, and
therefore the local views of process $i$ corresponding to $e$ and $e'$ are different in this case. So we assume now that 
$e$ belongs to some cube $\mathcal{Q}(\mathbf{x})$, and $e'$ to another subgraph, $\mathcal{Q}(\mathbf{x}')$. If $x_i \ne x'_i]$, or if the delivery patterns in $e$ and
$e'$ are different, then the local views of process $i$ are different, and we are done.
So suppose now, for contradiction, that (1)~process $i$ has the same local input in
$e\in E(\mathcal{Q}(\mathbf{x}))$ and $e'\in \mathcal{Q}(\mathbf{x}')$, (2)~$\varphi(\mathbf{v})=\varphi(\mathbf{w})$,  (3)~$\varphi(\mathbf{v}')=\varphi(\mathbf{w}'))$, 
and (4)~the local views of process $i$ in $e$ and $e'$
are the same. Let $S$ be the set of processes, excluding $i$, whose
messages are successfully transmitted under the delivery
pattern $\varphi(\mathbf{v})$. 
It must be the case
that the messages sent by processes in $S$
have the same value in $e$ and $e'$. As a consequence, $e$ and $e'$ agree
on the coordinates in $S$, and all other coordinates (excluding $i$)
must be~$0$, which implies that $e=e'$, a contradiction.
\end{proof}
\subsection{Characterizing Consensus Solvability Using the Reduced Kripke Graph}

Thanks to Lemma~\ref{lem:reduced-kripke-graph-preserve-information}, the error probabilities of consensus algorithms can be measured by using properties of specific edge-coloring of the reduced Kripke graph~$\mathcal{Q}$. 

\begin{lemma}\label{lem:characterization-by-2-coloring}
    There is a one-to-one correspondence between: 
    \begin{itemize}
        \item the set of 1-round full-information consensus algorithms satisfying the validity condition in the stochastic broadcast model with $n$ processes, and 
        \item the set of 2-edge-colorings of the $n$-dimensional reduced Kripke graph $\mathcal{Q}$ satisfying that all edges of $\mathcal{Q}(0^n)$ are colored~0, and  all edges of $\mathcal{Q}(1^n)$ are colored~1.
    \end{itemize}
\end{lemma}

\begin{proof}
Given an edge $e=\{\mathbf{v},\mathbf{v}'\}$ of $\mathcal{Q}$, of dimension~$i\in [n]$, let $\mathbf{x}\in\{0,1\}^n$ be such that $e \in E(\mathcal{Q}(\mathbf{x}))$. By item~1 of Lemma~\ref{lem:reduced-kripke-graph-preserve-information}, this edge can be mapped to a state of process~$i$ after one round of communication, as the local view of~$i$ is identical in the two delivery patterns, $\varphi(\mathbf{v})$ and $\varphi(\mathbf{v}')$, with input~$\mathbf{x}$. This mapping is well defined by item~2 of Lemma~\ref{lem:reduced-kripke-graph-preserve-information}. Finally, this mapping is injective by item~3 of Lemma~\ref{lem:reduced-kripke-graph-preserve-information}.

Conversely, after one round of communication, the state of each process~$i$ can be modeled as a vector $\mathbf{s}=(s_1,\dots,s_n)$ with $s_i=x_i$ and, for every  $j\neq i$, $s_j\in \{0,1,\star\}$ where, $s_j=0$ (resp., $s_j=1$) means that process $i$ has received the input~0 (resp.,~1) from process~$j$, and $s_j=\star$ means that process~$i$ has received no messages from process~$j$. Let us then define, for every $j\in [n]$, 
\[
h(s_j)=\begin{cases}
    -1 &\text{if $s_j=0$}\\
    1 &\text{if $s_j=1$}\\
    0 & \text{if $s_j=\star$}.
\end{cases}
\]
The function $h$ enables us to explicitly define the desired one-to-one correspondence between the set of possible states of the $n$ processes, and the edges of~$\mathcal{Q}$. Specifically, the state $\mathbf{s}$ of process $i$ is mapped to the edge $e$ between the vertices 
\[
\big(h(s_1),\dots,h(s_{i-1}),0,h(s_{i+1}),\dots,h(s_n)\big) 
\;\text{and}\;
\big(h(s_1),\dots,h(s_{i-1}),h(x_i),h(s_{i+1}),\dots,h(s_n)\big) 
\]
of~$\mathcal{Q}$. Let us denote this edge $e=h(i,\mathbf{s})$. This one-to-one correspondence between the local states of the processes and the edges of~$\mathcal{Q}$, allows us to infer a correspondence between binary consensus algorithms and 2-edge-coloring of~$\mathcal{Q}$. Given $\alg$, if process~$i$ outputs $y_i\in\{0,1\}$ whenever in state $\mathbf{s}$ after one round, then the edge $h(i,\mathbf{s})$ is colored~$y_i$. Conversely, given a 2-edge-coloring $c:E(\mathcal{Q}\to\{0,1\}$, if the edge 
$e=\{\mathbf{v},\mathbf{v}'\}\in\mathcal{Q}(\mathbf{x})$ is colored $b=c(e)$, where $\mathbf{v}=(v_1,\dots,v_n)$, and $\mathbf{v}'=(v'_1,\dots,v'_n)$ with $|v_i-v'_i|=1$ and $v'_j=v_j$ for all $j\neq i$, then process $i$ in state $(h^{-1}(v_1),\dots,h^{-1}(v_{i-1}),x_i,h^{-1}(v_{i+1}),\dots,h^{-1}(v_n))$ outputs~$b$. 
\end{proof}
Thanks to Lemma~\ref{lem:characterization-by-2-coloring}, we can interchangeably refer to consensus algorithms or 2-edge-coloring of~$\mathcal{Q}$. Given input $\mathbf{x}\in\{0,1\}^n$, each of the $2^n$ delivery patterns $\phi\in\{\bot,\top\}^n$ appears as a single vertex $v_\phi$ of~$\mathcal{Q}(\mathbf{x})$. This vertex has a  probability $\Pr(\phi)=p^kq^{n-k}$ associated with it, where $k$ denotes the number of occurrences of $\top$ in~$\phi$. Moreover, if an algorithm colors the incident edges of $v_\phi$ with different colors, then the corresponding algorithm fails to solve consensus for input $x$ and delivery pattern~$\phi$. This motivates the following definition. 

\begin{definition}
Let $\col(\mathcal{Q})$  denote the set of all colorings $c:E(\mathcal{Q})\to \{0,1\}$ of the edges of $\mathcal{Q}$ with colors in $\{0,1\}$ such that, for  $b=0,1$, we have $c(e)=b$ for all edges $e\in E(Q_{b^n})$. 
Given $c\in\col(\mathcal{Q})$, and for every subgraph $\mathcal{Q}(\mathbf{x})$ of $\mathcal{Q}$, the \emph{cut of $\mathcal{Q}(\mathbf{x})$ induced by $c$},
 is the set of vertices of $\mathcal{Q}(\mathbf{x})$ incident to an edge of $\mathcal{Q}(\mathbf{x})$ colored~0, and to an edge of $\mathcal{Q}(\mathbf{x})$ colored~1. Formally:
$$
\cut_c(\mathcal{Q}(\mathbf{x}))\DEF\Set{v\in V(\mathcal{Q}(\mathbf{x}))
\mid\exists u,w\in V(\mathcal{Q}(\mathbf{x}))\mbox{ s.t.~} c(v,u)\ne c(v,w)}
$$
\end{definition}

Note that, for every $c\in\col(\mathcal{Q})$, and every $b\in \{0,1\}$, $\cut_c(\mathcal{Q}(b^n))=\varnothing$. Note also that a vertex $\mathbf{v}\in\{-1,0,1\}^n$ of $\mathcal{Q}$ may be in $\cut_c(\mathcal{Q}(\mathbf{x}))$ for some $\mathbf{x}\in\{0,1\}^n$, but not in $\cut_c(\mathcal{Q}(\mathbf{x}'))$ for some other $\mathbf{x}'\in\{0,1\}^n$. We can now formulate the main result of this section, which is a graph-theoretical characterization of the  algorithms minimizing the error probability of solving consensus. 

\begin{theorem}
\label{thm:opt-error-proba-genreal}
    For every $p\in [0,1]$, the optimal error probability of any consensus algorithm in an $n$-process system with delivering probability~$p$ is  
 \(\displaystyle
\min_{c\in\col(\mathcal{Q})} \; \max_{\mathbf{x}\in\{0,1\}^n} \; \sum_{\mathbf{v}\in\cut_c(\mathcal{Q}(\mathbf{x}))}\Pr[\mathbf{v}].
\)
\end{theorem}

\begin{proof}
The specification of consensus is that all processes output
the same value. Each execution for an input $x$ 
corresponds to a delivery pattern $\phi$, and thus to the vertex $\mathbf{v}=g(\mathbf{x},\phi)$ of $\mathcal{Q}(\mathbf{x})$, which is associated to the  probability $\Pr[\mathbf{v}]=\Pr[\varphi(\mathbf{v})]=\Pr(\phi]$ determined by
the delivery pattern~$\phi$. The algorithm is successful for the pair $(\mathbf{x},\phi)$ if and only if all edges incident to $\mathbf{v}$ in $\mathcal{Q}(\mathbf{x})$ are colored using the same color.
\end{proof}

Note that, for every $\mathbf{x}\in\{0,1\}^n\smallsetminus \{0^n,1^n\}$, the center vertex $\mathbf{v}=(0,\dots,0)$ of $\mathcal{Q}$ belongs to $\cut_c(\mathcal{Q}(\mathbf{x}))$ for all $c\in\col(\mathcal{Q})$. Indeed, for every such~$\mathbf{x}$, the cube $\mathcal{Q}(\mathbf{x}))$ contains an edge from $\mathcal{Q}(0^n))$, and an edge from~$\mathcal{Q}(1^n))$. So, in particular, no consensus algorithm can achieve an error probability smaller than  $q^n=(1-p)^n$, for every~$p\in[0,1]$.

%%%%%%%%%%%%%%%%%%%%%%%%%%%%%%%%%
\subsection{1-Round  Algorithms Revisited}
%%%%%%%%%%%%%%%%%%%%%%%%%%%%%%%%%

It is interesting to apply the ``edge coloring perspective'' of Theorem~\ref{thm:opt-error-proba-genreal} to the consensus algorithms presented in Section~\ref{sec:1-round-algo-operational}.
Recall that, in Algorithm $\pref1$, processes output $1$ unless all the values they know
of are $0$. This means that all edges are colored $1$ except for the edges of $Q_{0^n}$.
Under this coloring, all vertices of the cube $Q_{0^n}$, with the exception of the vertex $\top^n$,
are cut vertices in the other cubes they belong to (the $\top^n$ vertices are ``corner vertices''
that belong to a single cube). As can be seen in \figref{fig-pref-courteous-hyperplane}, under this algorithm, 
the cut in some cubes is a full face, and in others it is an edge.

\begin{figure}[ht]
\centering
\includegraphics%[scale=.2]
[width=4.5cm,height=4cm]{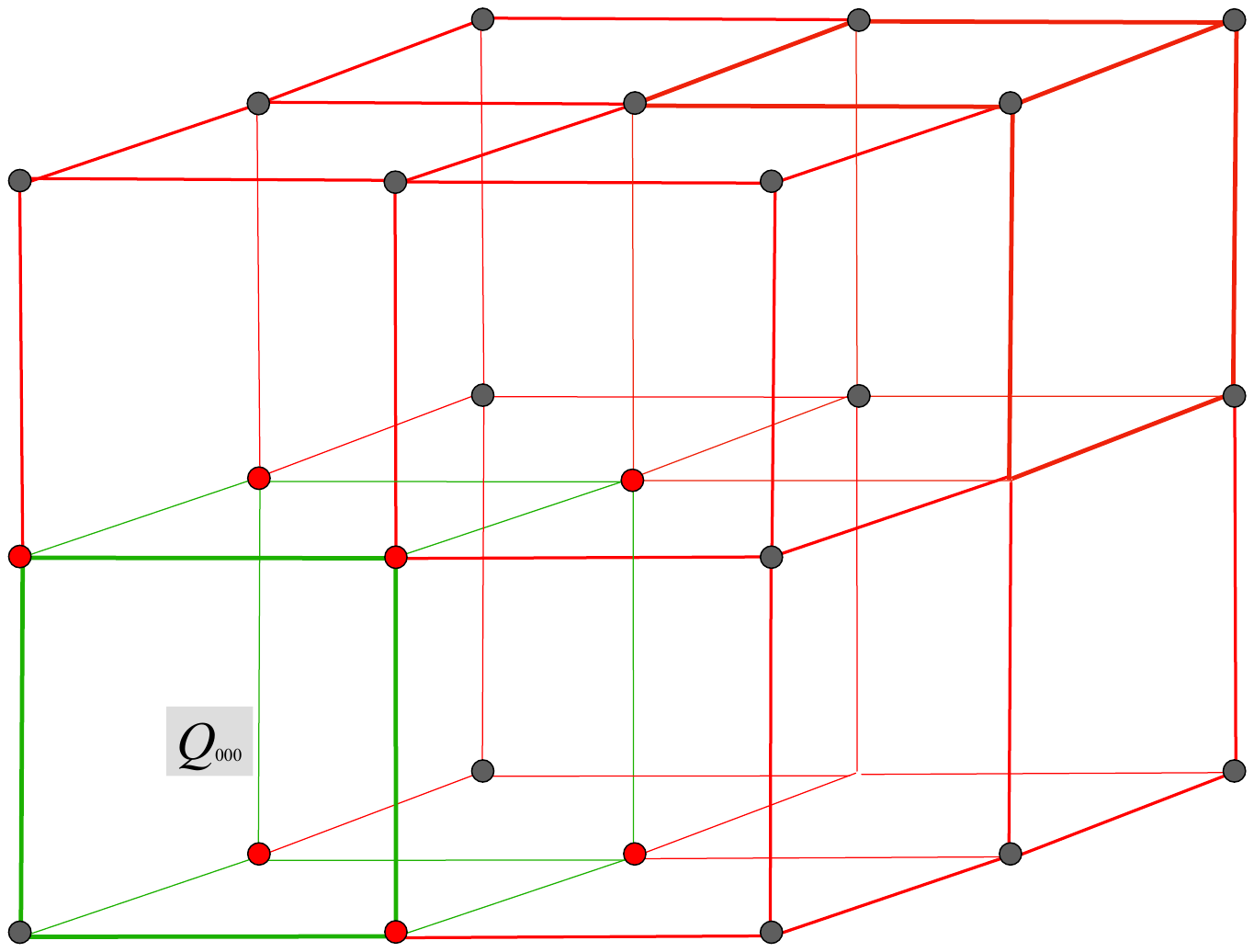}
\includegraphics%[scale=.2]
[width=4.5cm,height=4cm]{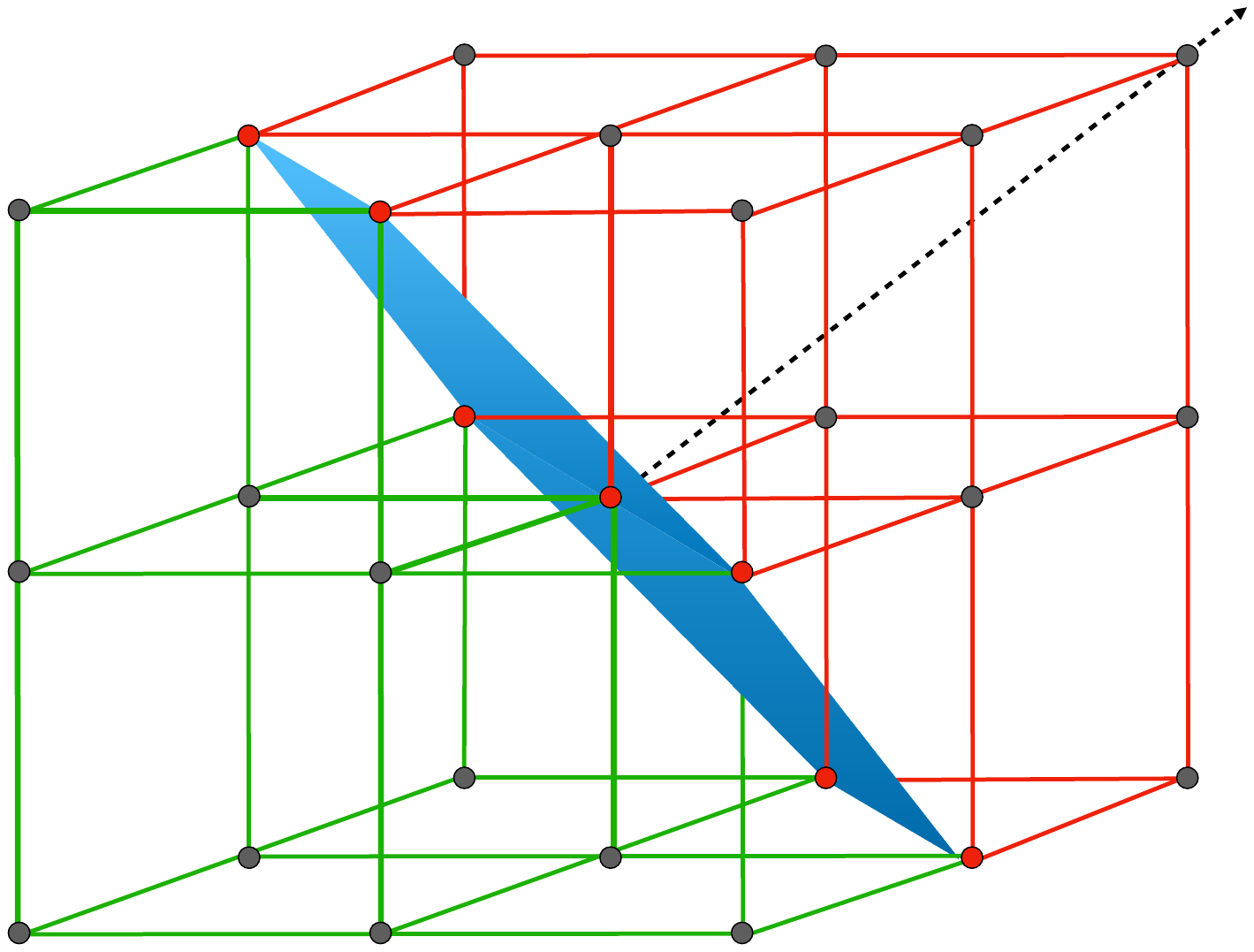}
\caption{(Left) Algorithm $\pref1$ colors all edges 1, except for the edges of $Q_{0^n}$. The figure
shows coloring for $n=3$. (Right) Algorithm $\cort$ coloring is determined by the hyperplane perpendicular
to the vector $1^n$. This means that the vertices in the cutset are exactly the vertices whose sum of 
coordinates is $0$. }
\label{fig-pref-courteous-hyperplane}
\end{figure}

Since $\court$ 
fails only when the number of successful transmissions of $1$ equals the number of successful transmissions of 
$0$, the cut vertices are exactly the vertices in which the sum of coordinates equals $0$.
Therefore, the edge colors are fixed in each of the two half-spaces defined by the hyperplane
perpendicular to the vector $1^n$. See \figref{fig-pref-courteous-hyperplane} for a
representation of the $n=3$ case.

%%%%%%%%%%%%%%%%%%%%%%%%%%%%%%%
\subsection{The Lower Bound}
%%%%%%%%%%%%%%%%%%%%%%%%%%%%%%%

Analyzing the properties of the Reduced Kripke graph, we show that
$\court$ and $\pref0$ achieve the optimal error probability (in
different ranges of $p$). More precisely, we prove the following.

\begin{theorem}\label{thm:lower-bound}
  Every 1-round consensus algorithm for a 3-process stochastic system
  satisfies
  \[
    \Pr[\mbox{error}]\geq
    \begin{cases}
      2p^2q+q^3 & \mbox{if $p\leq \nicefrac23$}\\
      q         & \mbox{otherwise.}
    \end{cases}
  \]
\end{theorem}

\begin{table}\centering
\begin{tabular}{|l|c|l|}\hline
 polynomials of the form& count& argument\\
 $a\cdot p^3+b\cdot p^2q+c\cdot pq^2+q^3$&&\\
\hline\hline
$a+c>0$ and $b>1$& 21& algebraic: each polynomial is larger than $2p^2+q^3$ \\
$a=1,b=1,c>1$& & for $p\in[0,2/3]$, and larger than $q$ for $p\in[2/3,1]$\\
$a=0,~b=1,~c=3$ & & \\
$a=1,~b=0,~c>1$ & & \\
$a=1,~b=1,~c=1$             & &  \\
\hline
$a=0,~b=0,~c=0$&7& topological: these polynomials cannot correspond to\\
$a=1,~b=0,~c=0$&& the error probability of a protocol\\
$a=0,~b=1,~c=0$&&\\
$a=1,~b=1,~c=0$&&\\
$a=1,~b=0,~c=0$&& \\
$a=0,~b=0,~c>1$&& \\
\hline
$a=0,~b=0,~c=1$&2& topological + indistinguishability: considering an adjacent  \\
$a=1,~b=0,~c=1$&&cube, corresponding to a different input assignment\\
\hline
\end{tabular}
    \caption{A classification of the arguments used in the proof. We rule  out 
    30 possible polynomials as representing the error probability of
    protocols better than the polynomials corresponding to $\court$ and $\pref0$. 
    }
    \label{tab:proof}
\end{table}

\theoremref{thm:lower-bound} intuitively, says that $\cort$ is optimal for $p\leq \nicefrac23$,
while algorithms such as $\pref_0$ and $\pref_1$ are all optimal for
$p\geq \nicefrac23$. The proof heavily uses the reduced Kripke
graph~$\mathcal{Q}$ (see
Figure~\ref{fig:reduced-Kripke-graph-mainText}). Recall that $\cort$
has error probability $2p^2q+q^3$ whereas $\pref_0$ and $\pref_1$ have
error probability~$q$.  More specifically, the
idea of the 
proof is as follows. 
We consider all possible error
polynomials, and we show that the above two are the smallest ones (in
their respective range) among all polynomials that correspond to protocols. Some of the polynomials  are ruled out  for
algebraic reasons: they are larger
than the two polynomials above in the relevant value ranges of $p$
(see
\tableref{tab:proof}). To disqualify the other polynomials, we
interpret them as vertices in some non-trivial cube of the Reduced Kripke
Graph (a cube that corresponds to a non-trivial input assignment),
and then show that these vertices cannot 
constitute a vertex cut in that cube, relying on the theory 
developed in this section.
We are then left with only two polynomials, which
require another argument, because
they correspond to protocols that are actually better than
$\court$ in the given cube. 
However, we can show that these protocols
are worse than
$\court$ for other inputs: the coloring induced by these protocols
implies a worse error probability in a certain adjacent cube (due to
the face it shares with the cube we consider). More intuitively, the reason is that while there
are protocols that are better than $\court$ for each input, they must
be worse than $\court$ on other inputs ($\court$ has the
same error probability for all non-trivial input assignments). We note that
appealing to an adjacent cube is using implicitly an indistinguishability
argument: saying that common edges must have the same color in both cubes
is tantamount to saying that a process must  make the same output on the same view under both input
assignments.

\begin{proof}[Proof of Theorem~\ref{thm:lower-bound}.] By \lemmaref{lem:characterization-by-2-coloring},
any consensus algorithm produces an edge-coloring ${c:E(\mathcal{Q})\to \{0,1\}}$  of the reduced Kripke graph~$\mathcal{Q}$, and conversely. To respect validity, the coloring must satisfy that, for every $b\in \{0,1\}$, and for every edge $e\in E(Q_{bbb})$, $c(e)=b$. That is, for input $000$ (resp., $111$) the output of every process must be~0 (resp.,~1). Given a coloring ${c:E(\mathcal{Q})\to \{0,1\}}$, and given an input $(x,y,z)\in\{0,1\}^3$, the probability of error of the algorithm defined by~$c$ for input $(x,y,z)$, can be merely computed by inspecting the vertices of the cube $Q_{xyz}$ of~$\mathcal{Q}$. Each vertex of $Q_{xyz}$ is a vertex $v_\phi$ for the failure pattern $\phi\in\{\bot,\top\}^3$. The outputs of the three processes for $\phi$ is the multiset of colors of the edges of $Q_{xyz}$ incident on~$v_\phi$. The coloring does not err for input  $(x,y,z)$ and failure pattern~$\phi$ if all edges incident on $v_\phi$ have the same color. Moreover, each failure pattern~$\phi$ has a probability of occurring, which depends on~$p$. Therefore, for input  $(x,y,z)$, the coloring errs with probability $\sum \Pr[v_\phi]$ where the sum is taken over all vertices $v_\phi$ that are bichromatic, i.e., for which not all incident edges are of the same color.

\begin{figure}[htb]
\centering
\includegraphics[scale=.8]{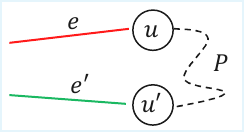}
\caption{Illustration for \lemmaref{lem:each-path-has-a-cut}.
$e$ and $e'$ may share their left enpoint, but $u\ne u'$.}
\label{fig:lem6-2}
\end{figure}

The following simple lemma plays an important role in our analysis. 

\begin{lemma}\label{lem:each-path-has-a-cut}
Let $c:E(\mathcal{Q})\to \{0,1\}$ be an edge-coloring of $\mathcal{Q}$, and let $G$ be a subgraph of~$\mathcal{Q}$. Let $e=uv$ and $e'=u'v'$ be two edges of $G$ with $u\neq u'$ and $c(e)\neq c(e')$. Let $P$ be a path in $G$ with two extremities  $u$ and $u'$, not containing $v$ nor~$v'$. At least one vertex of $P$ belongs to $\cut_c(G)$. 
\end{lemma} 

\begin{proof}
See \figref{fig:lem6-2}.
If no vertices of $P$ are in $\cut_c(G)$, then all its edges are of the same color $b\in\{0,1\}$. Since $c(e)\neq c(e')$, let us assume, w.l.o.g., that $b=c(e)$. It follows that $u'$ is incident to an edge of~$P$, with color $b=c(e)$, and is also incident to the edge~$e'$, with color $c(e')\neq c(e)=b$. As a consequence, $u'\in \cut_c(G)$, in contradiction with the assumption that  no vertices of $P$ are in~$\cut_c(G)$. 
\end{proof}

Another simple lemma plays an important role in our analysis, which can be checked by a mere inspection of  Figure~\ref{fig:reduced-Kripke-graph-mainText}. 

\begin{lemma}\label{lem:one-face-and-one-edge}
    For every $(x,y,z)\in\{0,1\}^3\smallsetminus \{(0,0,0),(1,1,1)\}$, there exists a unique $b\in\{0,1\}$ such that the cube $Q_{xyz}$ of $\mathcal{Q}$ shares exactly one 2-dimensional face $F$ with $Q_{bbb}$, and exactly one edge $e$ with $Q_{\bar b \bar b \bar b}$.  
\end{lemma}

Note that all the faces and edges that are shared between a cube $Q_{xyz}$ and the cubes $Q_{000}$ and $Q_{111}$ include the center vertex $v_\phi$ of~$\mathcal{Q}$ corresponding to the failure pattern $\phi=\bot\bot\bot$. Using the notations in the statement of Lemma~\ref{lem:one-face-and-one-edge}, we have $V(e)\cap V(F)=v_\phi$. 

As a quick warmup, let us start by showing that, for every $p\leq \nicefrac13$, $\cort$ is optimal.%
\footnote{
This does not cover the range $p\in (\frac13,\frac23]$ where Theorem~\ref{thm:lower-bound} claims that $\cort$ is also optimal, but it gives a good intuition of the proof, and provides us with arguments that can be reused later in the proof. 
}

\begin{figure}[htb]
\centering
\includegraphics[scale=.3]{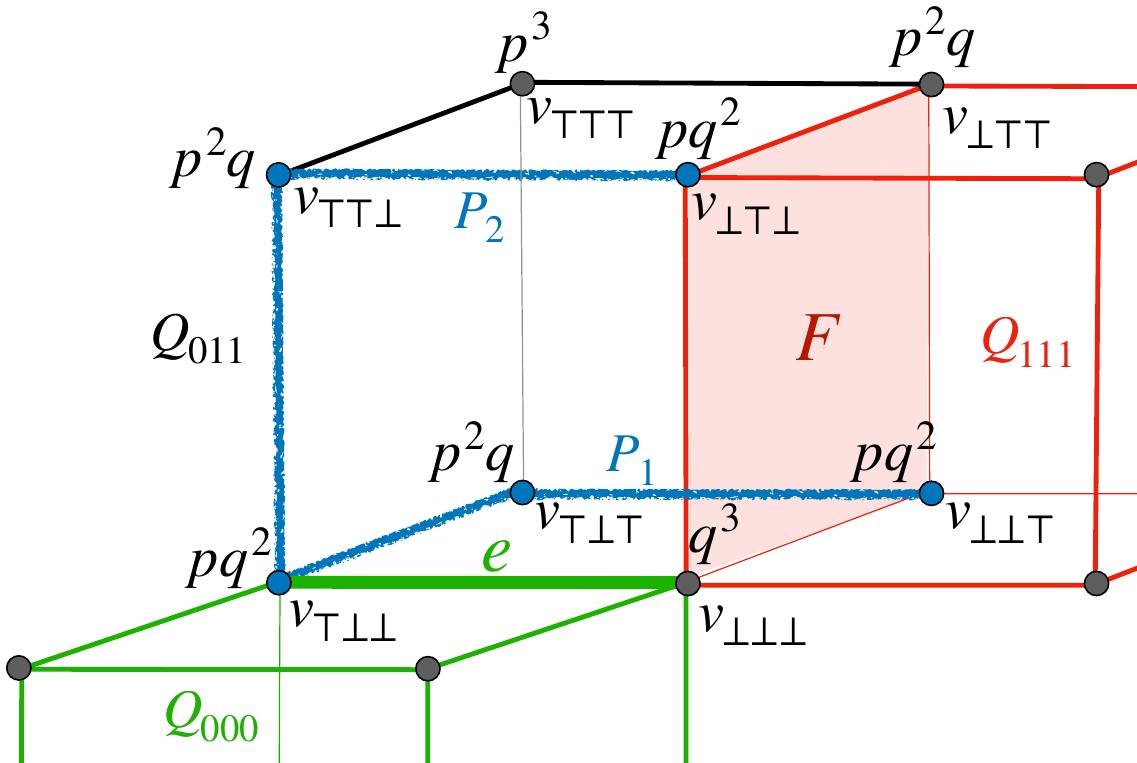}
\caption{Construction in the proof of Proposition~\ref{prop:lwb-simple}. The cube $Q_{011}$ shares the face $F$ with $Q_{111}$, and the edge $e$ with $Q_{000}$. The validity condition imposes that all edges of $F$ are colored~1, whereas the edge $e$ must be colored~0. }
\label{fig:lwb-corteous-p-small}
\end{figure}

\begin{proposition}\label{prop:lwb-simple}
For $p\leq \nicefrac13$, every 1-round consensus algorithm for a 3-process stochastic system satisfies $\Pr[\mbox{error}]\geq 2p^2q+q^3$. 
\end{proposition}

\begin{proof}
By Lemma~\ref{lem:one-face-and-one-edge}, and as illustrated in Figure~\ref{fig:lwb-corteous-p-small}, the cube $Q_{011}$ is sharing the face $F=\{v_{\bot\bot\bot},v_{\bot\top\bot},v_{\bot\top\top},v_{\bot\bot\top}\}$ with~$Q_{111}$, and the edge $e=\{v_{\bot\bot\bot},v_{\top\bot\bot}\}$ with~$Q_{000}$. 
Let us thus consider the two paths 
\[
P_1=(v_{\bot\bot\top},v_{\top\bot\top},v_{\top\bot\bot}) \;\mbox{and}\; P_2=(v_{\bot\top\bot},v_{\top\top\bot},v_{\top\bot\bot})
\]
in  $Q_{011}$, as displayed in  Figure~\ref{fig:lwb-corteous-p-small}.  Both $v_{\bot\bot\top}$ and $v_{\bot\top\bot}$ have incident edges colored~1, while $v_{\top\bot\bot}$ has an incident edge colored~0. Therefore, thanks to Lemma~\ref{lem:each-path-has-a-cut} applied with $G=Q_{011}$, each of the two paths $P_1$ and $P_2$ must have a vertex in~$\cut_c(Q_{011})$. 
If any of the three vertices of $P_1\cup P_2$ with probability $pq^2$ is in $\cut_c(Q_{011})$,  then 
$
\sum_{v_\phi\in\cut_c(Q_{011})}\Pr[v_\phi]\geq pq^2+q^3.
$
For $p\leq 1/3$, we have $pq^2\geq 2p^2q$, and thus
$
\sum_{v_\phi\in\cut_c(Q_{011})}\Pr[v_\phi]\geq  2p^2q+q^3, 
$ 
as claimed. On the other hand, if none of these three vertices are in~$\cut_c(Q_{011})$, then,  thanks to Lemma~\ref{lem:each-path-has-a-cut}, it must be the case that both $v_{\top\bot\top}$ and $v_{\top\top\bot}$ are in the cut. As a consequence, 
$
\sum_{v_\phi\in\cut_c(Q_{011})}\Pr[v_\phi]\geq 2p^2q+q^3 
$ 
in this case as well, which completes the proof. 
\end{proof}

\subparagraph{Remark.}

The proof of Proposition~\ref{prop:lwb-simple} is ``local'' in the sense that it involves only one cube $Q_{xyz}$ with $(x,y,z)\notin\{(0,0,0),(1,1,1)\}$. It follows that the lower bound $\Pr[\mbox{error}]\geq 2p^2q+q^3$ for $p\leq \nicefrac13$ is not only a worst-case lower bound, but it also holds for each and every input $(x,y,z)\notin\{(0,0,0),(1,1,1)\}$. (The proof is given for $Q_{011}$ but it holds for every $Q_{xyz}$, by mere symmetry.) In other words, for any non-trivial input $(x,y,z)$, any consensus algorithm errs with probability at least $2p^2q+q^3$ on that input, whenever $p\leq \nicefrac13$. In other words, for $p\leq \nicefrac13$, $\cort$ is optimal  for every input as far as the error probability is concerned.
\medskip

At this point, we have all the tools for dealing with the general case $p\in[0,1]$. Every cube $Q_{xyz}$ has one vertex corresponding to a failure pattern occurring with probability~$p^3$, three vertices corresponding to a failure pattern occurring with probability~$p^2q$, three vertices corresponding to a failure pattern occurring with probability~$pq^2$, and one vertex corresponding to a failure pattern occurring with probability~$q^3$. It is therefore convenient to express the probability of failure of an algorithm, or equivalently of an edge-coloring of~$\mathcal{Q}$, as a degree-3 polynomial with two variables, $p$ and $q=1-p$, of the form 
\begin{equation}\label{eq:poly-f-abc}
    f_{a,b,c}(p)=a \, p^3 + b \, p^2q +c \, pq^2 + q^3
\end{equation}
where the three integers $a,b$, and $c$ satisfy: 
\[
0\leq a \leq 1, \;\;\;\; 0\leq b \leq 3, \;\; \mbox{and} \;\; 0\leq c \leq 3.
\]
In particular, the error probability of $\cort$ is captured by the polynomial 
\[
f_{0,2,0}(p)=2p^2q+q^3,
\]
and the error probability of, e.g., $\pref1$ is captured by the polynomial 
\[
f_{0,1,2}(p)=p^2q+2pq^2+q^3=q(p+q)^2=q.
\]
This yields 30 other  forms of error probabilities. However, many can be easily ruled out. This is, for instance, the case of the seven polynomials $f_{a,2,c}$ with $a+c>0$, and of the eight polynomials $f_{a,3,c}$ with $a+c\geq 0$, by direct comparison to the error probability of $\cort$. This is also the case for the three polynomials $f_{1,1,2}$, $f_{0,1,3}$, and $f_{1,1,3}$ by direct comparison with the error probability of $\pref1$. It also holds that, for every $p\in[0,1]$, $f_{1,1,1}(p)\geq 2p^2q+q^3$, and $f_{1,0,3}(p)\geq f_{1,0,2}(p)\geq 2p^2q+q^3$, which rules out another group of three polynomials. 

\begin{figure}[htb]
\centering
\includegraphics[scale=.3]{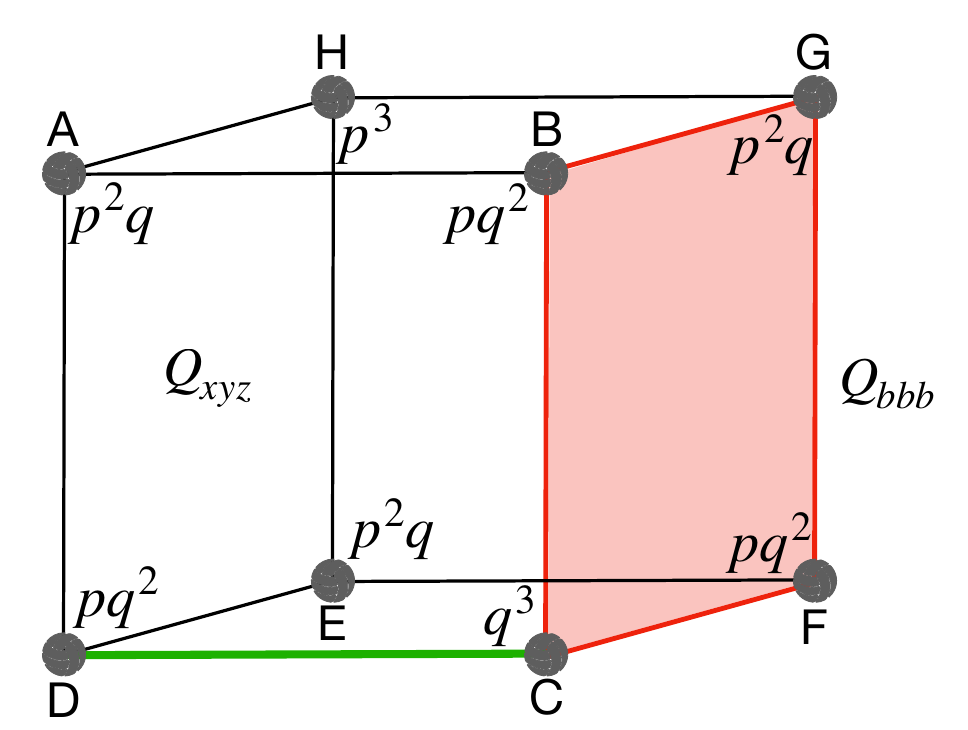}
\caption{Local arguments in the proof of Theorem~\ref{thm:lower-bound}. The ``non-trivial'' cube $Q_{xyz}$, i.e., $(x,y,z)\notin\{(0,0,0),(1,1,1)\}$, shares the marked face with $Q_{bbb}$, and the  edge $\{C,D\}$ with $Q_{\bar b,\bar b,\bar b}$, where $b\in\{0,1\}$. }
\label{fig:lower-bound-local}
\end{figure}

Other polynomials can be eliminated by using similar arguments as in the proof of Proposition~\ref{prop:lwb-simple}. Let us consider Figure~\ref{fig:lower-bound-local}, which displays a cube $Q_{xyz}$ sharing a 2-dimensional face $\{B,C,F,G\}$ with $Q_{bbb}$, and one edge $\{C,D\}$ with $Q_{\bar b,\bar b,\bar b}$, for some $b\in\{0,1\}$ --- by Lemma~\ref{lem:one-face-and-one-edge}, such a scenario holds for every ``non-trivial'' cube $Q_{xyz}$. Lemma~\ref{lem:each-path-has-a-cut} states that, for every coloring~$c$, each path from any of the three vertices $B,F$, or $G$ to $D$ must include a vertex in $\cut_c(Q_{xyz})$. This directly eliminates the polynomial $f_{0,0,0}$ as it corresponds to a coloring with no vertex cuts (but~$C$). Lemma~\ref{lem:each-path-has-a-cut} also rules out $f_{1,0,0}$ as it corresponds to a coloring with only $H$ (and $C$) in the cut. Similarly, 
$f_{0,1,0}$ (which corresponds to $C$, and $A,E$ or $G$ in the cut), and 
$f_{1,1,0}$ (which corresponds to $C,H$, and $A,E$ or $G$ in the cut) can also be eliminated by applying Lemma~\ref{lem:each-path-has-a-cut}. 

Unfortunately, $f_{0,0,1}$ cannot be discarded by such ``local arguments.'' Indeed, $C$ and $D$ may be the only cut vertices. For instance, in $Q_{011}$ (see Figure~\ref{fig:lwb-corteous-p-small}), $v_{\bot\bot\bot}$ and $v_{\top\bot\bot}$ may be the only vertices in the cut. This is, for instance, the cut that would result from Algorithm $\pref_1$. So, we are left with five polynomials, which are 
\[
f_{0,0,1}, 
f_{0,0,2}, 
f_{0,0,3},
f_{0,1,1},
\; \mbox{and} \;
f_{1,0,1}.  
\]
Interestingly, while $f_{0,0,1}$ cannot be eliminated by a local argument using Figure~\ref{fig:lower-bound-local}, such an argument can still be used to rule out $f_{0,0,2}$, $f_{0,0,3}$, and $f_{1,0,1}$ from consideration. 
\begin{itemize}
    \item Indeed, $f_{0,0,3}$ corresponds to the four vertices $B,C,D,F$ in the cut. Having $B$ in the cut means that the edge $\{A,B\}$ has a color different from $\{B,G\}$. This leaves the path $(A,H,G)$ without any vertex in the cut, contradicting Lemma~\ref{lem:each-path-has-a-cut}. 
    
    \item Eliminating $f_{0,0,2}$ requires considering two cases. If $D$ is not in the cut (but only $C,B$, and~$F$), then there are no cut vertices along the path $(D,A,H,G)$, contradicting Lemma~\ref{lem:each-path-has-a-cut}. If $D$ is in the cut, then let us assume, w.l.o.g., that $B$ and $C$ are the other two vertices in the cut (the case $C$ and $F$ is identical by symmetry). As $B$ is in the cut, the edge $\{A,B\}$ has a color different from the edge $\{B,G\}$, but there are no cut vertices along the path $(A,H,G)$, contradicting Lemma~\ref{lem:each-path-has-a-cut}. 

    \item Finally, $f_{1,0,1}$ corresponds to vertices $C$ and $H$ in the cut, plus one of the three vertices $B,D$, and~$F$. Since $H$ is in the cut, the edge $\{G,H\}$ is of a color different from $\{C,D\}$, and at least one of the other two edges incident to $H$ has a color equal to the color of $\{C,D\}$. By symmetry, one can assume w.l.o.g. that $\{A,H\}$ has the same color as $\{C,D\}$. By Lemma~\ref{lem:each-path-has-a-cut}, this forces $B$ to be in the cut, which must thus be $\{B,C,H\}$. This leaves the path $(D,E,F)$ free of any cut vertices, contradicting Lemma~\ref{lem:each-path-has-a-cut}.
\end{itemize}

\begin{figure}[htb]
\centering
\includegraphics[scale=.3]{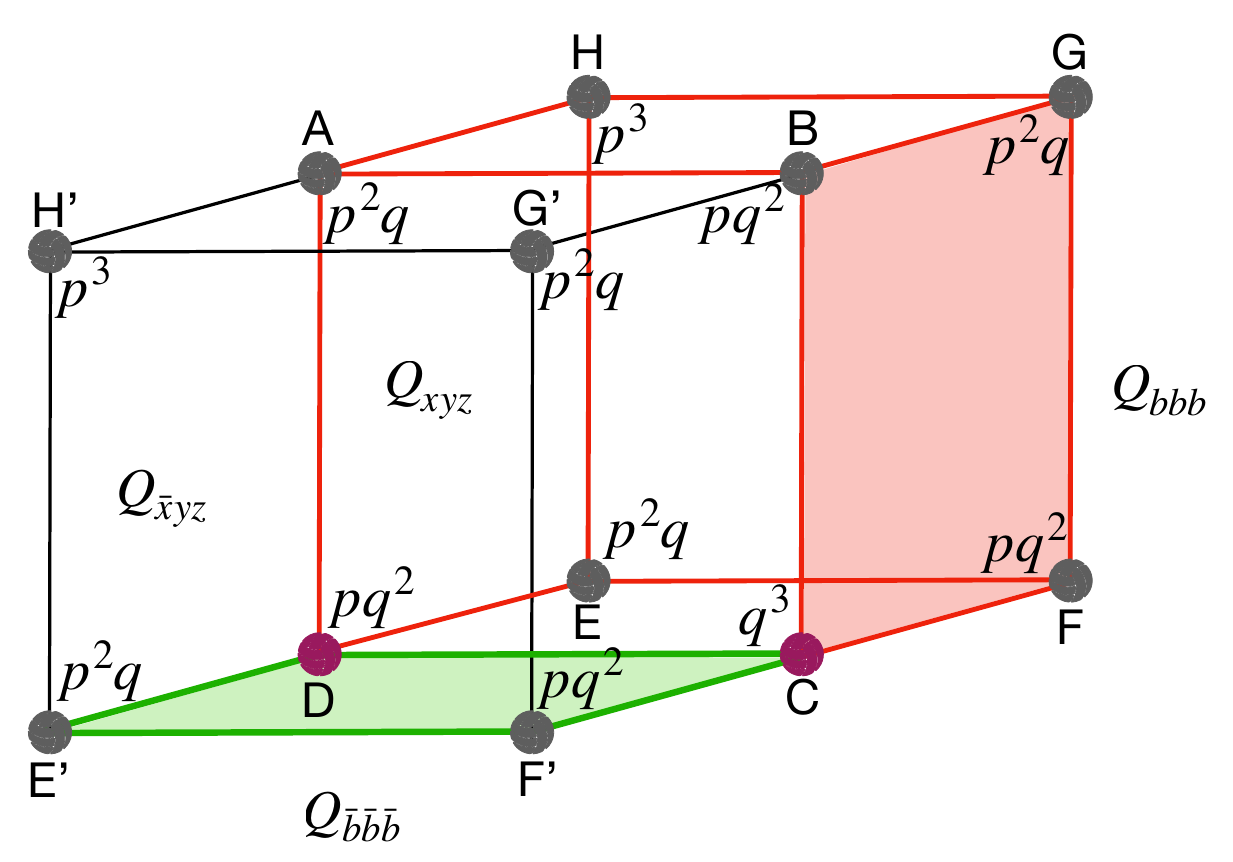}
\caption{Global arguments in the proof of Theorem~\ref{thm:lower-bound}, which proceeds by considering a ``nontrivial'' cube sharing the edge $\{C,D\}$ with $Q_{xyz}$, say $Q_{\bar xyz}$. The latter shares the marked face with $Q_{\bar b\bar b\bar b}$, and the edge $\{B,C\}$ with $Q_{bbb}$.}
\label{fig:lower-bound-global}

\end{figure}

So now  only two polynomials remain: $f_{0,0,1}$, and~$f_{0,1,1}$. As already mentioned before, these two polynomials cannot be discarded by considering a single input. Indeed, for the single input $(x,y,z)$, $f_{0,0,1}$ could correspond to the cut $\{C,D\}$, and $f_{0,1,1}$ could correspond to the cuts $\{A,C,D\}$ or $\{C,D,E\}$. In fact, these cuts are the only possible ones corresponding to $f_{0,0,1}$, and~$f_{0,1,1}$. Eliminating these two polynomials requires considering multiple inputs. Figure~\ref{fig:lower-bound-global} displays two cubes: $Q_{xyz}$ as in Figure~\ref{fig:lower-bound-local}, and an adjacent cube, say $Q_{\bar xyz}$, satisfying that 
\[
\{(x,y,z),(\bar x,y,z)\}\cap \{(0,0,0),(1,1,1)\}=\varnothing.
\]
The cube $Q_{xyz}$ shares the face $\{B,C,F,G\}$ with $Q_{bbb}$, while $Q_{\bar xyz}$ shares the face $\{C,D,E',F'\}$ with $Q_{\bar b\bar b\bar b}$. The cube $Q_{\bar xyz}$ shares the edge $\{B,C\}$ with $Q_{bbb}$. 

As stated above, the polynomial $f_{0,0,1}$  corresponds to the cut $\{C,D\}$ in $Q_{xyz}$. This implies that all edges of $Q_{xyz}$ but $\{C,D\}$ are colored~$b$. As a consequence, $C$ and $D$ are in the cut of $Q_{\bar xyz}$, which already amounts to $f_{0,0,1}(p)$. This cut is not sufficient to separate in $Q_{\bar xyz}$ the vertex $A$ incident to the edge $\{A,B\}$ colored $b$ from the vertex $E'$ incident to the edge $\{E',F'\}$ colored~$\bar b$, contradicting Lemma~\ref{lem:each-path-has-a-cut}. 

Finally, the polynomial $f_{0,1,1}$ could correspond to only two possible cuts in~$Q_{xyz}$: $\{A,C,D\}$ or $\{C,D,E\}$. The cut $\{C,D,E\}$ colors the edge $\{A,D\}$ by~$b$, forcing $A$ or $E$ to be in the cut of~$Q_{\bar xyz}$. In both cases, no vertices in the path $(F',G',B)$ are in the cut, contradicting Lemma~\ref{lem:each-path-has-a-cut}. We are left with the case where $\{A,C,D\}$ is the cut corresponding to $f_{0,1,1}$ in $Q_{xyz}$. This implies that $\{A,D\}$ is colored $\bar b$, and thus $A$ and $C$ are in the cut for $Q_{\bar xyz}$. Thanks to Lemma~\ref{lem:each-path-has-a-cut} applied to the paths $(E',H',G',B)$ and $(F',G',B)$, this forces $B$ to be in the cut of $Q_{\bar xyz}$. So the only possible cut for $Q_{\bar xyz}$ not exceeding $f_{0,1,1}$ is $\{A,B,C\}$, as depicted in Figure~\ref{fig:lower-bound-global2}. We thus now focus on the cube $Q_{\bar x\bar yz}$ satisfying  
\[
\{(x,y,z),(\bar x,y,z),(\bar x,\bar y,z)\}\cap \{(0,0,0),(1,1,1)\}=\varnothing.
\]
The two vertices $B$ and $C$ are in the cut of $Q_{\bar x\bar yz}$ (Figure~\ref{fig:lower-bound-global2}). By Lemma~\ref{lem:each-path-has-a-cut}, two more vertices must be added, one for the path $(G',X,A')$, and one for the path $(F',Y,D')$, which cannot be afforded without exceeding an error probability $f_{0,1,1}(p)$ for input $(\bar x,\bar y,z)$ because $p^3+2pq^2+q^3>f_{0,1,1}(p)$ for every $p\in[0,1]$. 

\begin{figure}[htb]
\centering
\includegraphics[scale=.3]{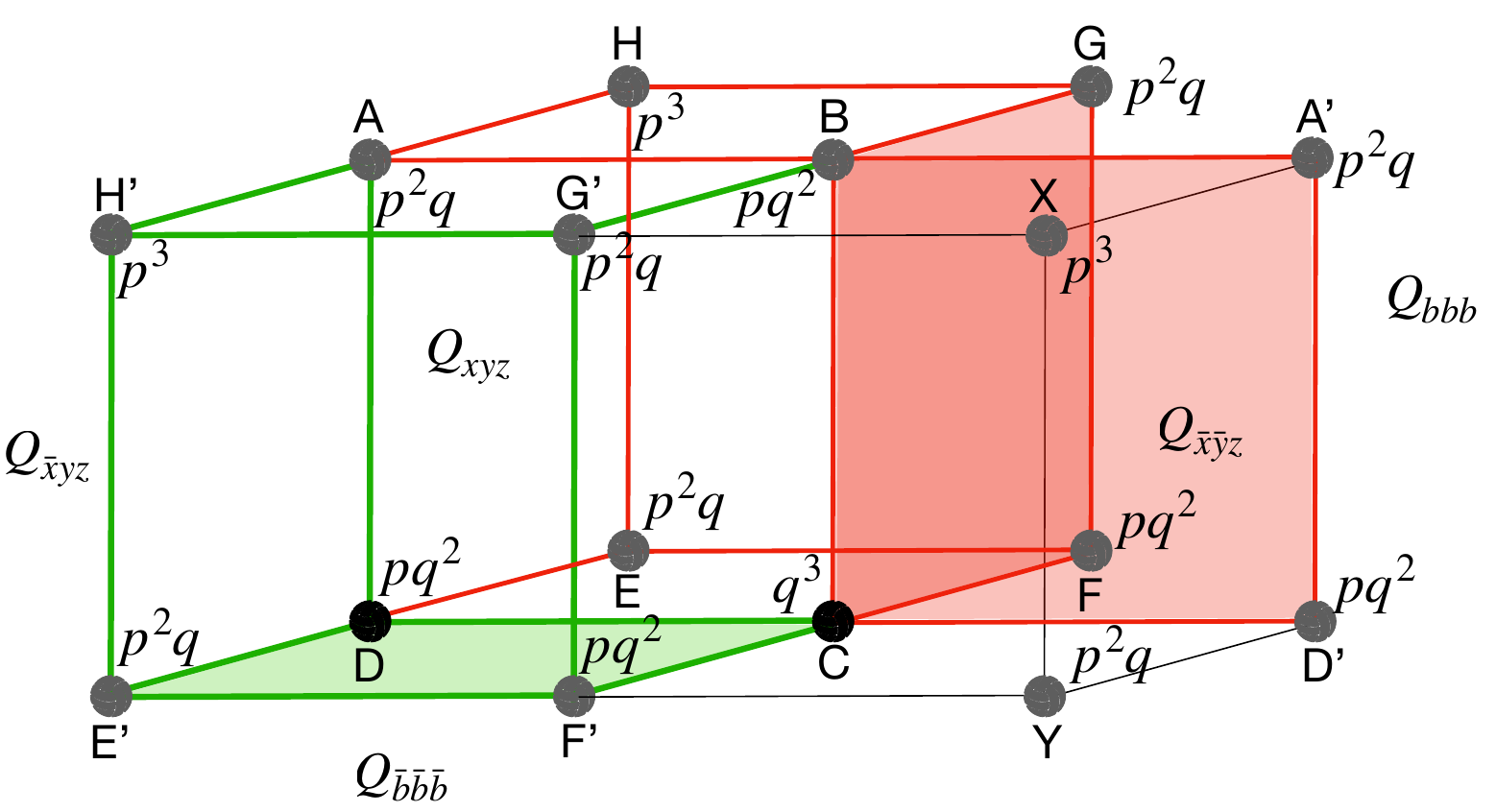}
\caption{Global arguments in the proof of Theorem~\ref{thm:lower-bound} continued, by considering a third non-trivial cube, sharing a face with both $Q_{bbb}$ and $Q_{\bar xyz}$, say $Q_{\bar x\bar yz}$, and the edge $\{C,F'\}$ with $Q_{\bar b\bar b\bar b}$. }
\label{fig:lower-bound-global2}
\end{figure}

\subparagraph{Summary.}

The analysis detailed in this section proves that, for every polynomial $f_{a,b,c}$ of the form defined in Eq.~\eqref{eq:poly-f-abc}, 
\begin{itemize}
    \item either $f_{a,b,c}(p)\geq \min\{f_{0,2,0}(p),f_{0,1,2}(p)\}$ for every $p\in [0,1]$ (i.e., the error probability $f_{a,b,c}(p)$ is not smaller than the error probabilities of $\cort$ and $\pref1$), 
    \item or, if there exists $p$ such that $f_{a,b,c}(p)< \min\{f_{0,2,0}(p),f_{0,1,2}(p)\}$ then 
    \begin{itemize}
        \item either there are no non-trivial inputs $(x,y,z)$ such that an algorithm can err with probability $f_{a,b,c}(p)$ for this input, 
        \item or, if there exists a non-trivial input $(x,y,z)$ such that an algorithm can err with probability $f_{a,b,c}(p)$ for this input, then there exists another non-trivial input $(x',y',z')$ such that the algorithm errs with probability larger than $f_{a,b,c}(p)$ for this other input. 
    \end{itemize}
\end{itemize}
In other words, $f_{0,2,0}(p)$ and $f_{0,1,2}(p)$ are the only possible forms of minimum error probability for an algorithm solving consensus. This completes the proof of Theorem~\ref{thm:lower-bound}. 
\end{proof}

%%%%%%%%%%%%%%%%%%%%%%%%%%%%%%%%%
\section{Upper and Lower Bounds for $n\geq 3$ Processes}
\label{sec:gen}
%%%%%%%%%%%%%%%%%%%%%%%%%%%%%%%%%

In this section we first
present an upper bound on the error probability
of multi-round algorithm for consensus with any number of processes.
While we do not know what is the best strategy in terms of
error probability for a given number of {rounds}, our algorithm gives
the best possible error probability for a given number of 
\emph{transmissions}.
Then we proceed to show that the error probabilities
of $\court$ and $\pref1$ are optimal for any number of 
processes in a single round---$\court$ is optimal for small values of $p$ and $\pref1$ is optimal for large values of $p$.

%%%%%%%%%%%%%%%%%%%%%%%%%%%%%%%%%
\subsection{Algorithm $\sweep$: Optimal Transmission Count}
%%%%%%%%%%%%%%%%%%%%%%%%%%%%%%%%%

It might be tempting to believe that optimal $r$-round consensus algorithms for a given probability $p$ of broadcast success can be merely obtained by repeating $r$ times an optimal 1-round algorithm. Specifically, after each application of the 1-round algorithm, each process takes the output produced by this algorithm as input for the next round. This approach does systematically yield optimal $r$-round consensus algorithms for the case of two processes~\cite{FraigniaudPSR25}. This is however not the case even for three processes. 

Indeed, consider the following $2r$-round $\sweep^r$  Algorithm~\ref{alg:sweep}. The algorithm consists of $r\ge1$ phases,
where a phase consists of $2$ rounds. In each phase, the
algorithm  first runs $\pref0$ and then $\pref1$.
The input to a round is the output of the previous round.

\begin{algorithm}[htb]
\caption{$\sweep^r$ --- Code of process $i\in [n]$.$r\in\mathbb{N}$ is a parameter}
\label{alg:sweep}
\begin{algorithmic}[1]
\State $v_i\gets \In_i$
\label{sweep1}
\For{$r$ iterations}
\State \textbf{if} $v_i=0$ \textbf{then} broadcast $0$ \label{sweep2}
\Comment{1st round}
\State receive messages (if any)
\State \textbf{if} $count_i(0)>0$ \textbf{then} $v_i\leftarrow 0$
\label{sweep4}
\State \textbf{if} $v_i=1$ \textbf{then} broadcast $1$  \label{sweep6}
\Comment{2nd round}
\State receive messages (if any)
\State \textbf{if} $count_i(1)>0$ \textbf{then} $v_i\leftarrow 1$ \Comment{the of messages received at step 6}
\label{sweep7}
\EndFor
\State $\Out_i\leftarrow v_i$
\label{sweep8}
\end{algorithmic}
\end{algorithm}

We show that $\sweep^r$ is optimal as far as the number of broadcast transmissions is concerned. A \emph{transmission} is merely defined as one broadcast operation from any process to all the other processes. The number of transmissions of an algorithm is the maximum, over all inputs and  all possible executions, of the number of broadcast operations invoked by the algorithm, regardless of whether they are  successful or not. 

\begin{theorem}
\label{thm:sweepr}
    For every number $n\geq 2$ of processes, and for every $r\geq 0$, $\Pr[\sweep^r\;\mbox{errs}]=q^{rn}$ and the number of transmissions of $\sweep^r$ equals~$rn$.
    This error probability is optimal 
    for $rn$ transmissions. 
\end{theorem}

\begin{proof}
First we show that in every pair of rounds 
(\linref{sweep2}---\linref{sweep7}), the total
number of broadcasts is at most $n$. To see that, 
let $c(0)$ (and $c(1)$) be the number of processes $i$ 
with $v_i=0$ (and $v_i=1$, respectively), before \linref{sweep2}. By the code, $c(0)$ processes 
transmit in \linref{sweep2}; if any of them is successful, 
no process transmits in \linref{sweep6}, and if no broadcast
of \linref{sweep2}
was successful, $c(1)$ processes transmit
in \linref{sweep6}. In any case, the total number 
transmissions in \linref{sweep2} and \linref{sweep6} is at most 
$c(0)+c(1)=n$.
Since $\sweep^r$ consists of $r$ iterations of such round pairs,
the total number of transmissions is at most $rn$. Moreover,
the total number of transmissions is strictly less than
$rn$ only if at least one broadcast was successful.

Next, we claim that if any broadcast is successful, then the protocol
terminates with agreement. To see that, suppose that the first
successful broadcast occurs at round $t$, where $1\le t\le 2r$.
By the specification of the algorithm, if $t$ is odd then the contents of the message is $0$, a
and if $t$ is even, the contents of the message is $1$. Let us
assume, without loss of generality, that $t$ is odd (the
argument for even $t$ is essentially identical). Since
the contents of the successful message is $0$ (possibly
multiple broadcasts succeed, but they all have the same contents ``$0$'').
Since the broadcast is successful, it is received by all processes,
and then, by the specification of the algorithm, all
processes $i\in[n]$ will have $v_i= 0$: the processes
with $\In_i0=$ have set it in \linref{sweep1} and 
didn't change it since, because this is the first round with
a successful broadcast; and the other processes set $v_i\gets 0$
in \linref{sweep4}. Consequently, since the algorithm 
uses only the $v_i$ variables, the value $1$ is effectively wiped
out of the system: all processes, after round $t$,
will keep silent on even rounds and broadcast ``$0$'' in odd rounds. This shows that the protocol fails (on a non-trivial
input) if
and only if all broadcasts fail. By the analysis above,
this occurs only if there were $rn$ transmissions,
and hence $\sweep^r$ fails with probability at most $q^{rn}$.

The optimality of $\sweep^r$ with respect to  number of transmissions follows from the observation that, for every $k\geq 0$, any consensus algorithm performing $k$ transmissions in the stochastic broadcast model has an error probability of at least~$q^k$ when the input is non-trivial. 
\end{proof}

\noindent Note that, for every $n\geq 2$, $r\geq 0$, and $p\in[0,1]$, $\sweep^r$ satisfies
\[
\Pr[\mathsf{OPT}(2r)\;\mbox{errs}] \leq 
\Pr[\sweep^r\;\mbox{errs}] \leq 
\Pr[\mathsf{OPT}(r)\;\mbox{errs}], 
\]
where, for any $t\geq 0$, $\mathsf{OPT}(t)$ is a $t$-round algorithm failing with minimum probability. The first inequality holds because $\sweep^r$ is a $2r$-round algorithm, and the second inequality holds because any $r$-round algorithm errs with probability at least~$q^{rn}$, which is the error probability of $\sweep^r$.  In other words, $\sweep^r$ is at least as good (in terms of error probability) as any 
algorithm performing half as many rounds. 

%%%%%%%%%%%%%%%%%%%%%%%%%%%%%%%%%
\subsection{Lower Bounds for Single-Round Algorithms for Small and Large $p$ values}
\label{sec:loweBoundsG}
%%%%%%%%%%%%%%%%%%%%%%%%%%%%%%%%%

In this section, we show that, for every $n\geq 3$, Algorithm $\cort$ is
optimal when the success probability $p$ of broadcast is close to~0, and that Algorithms  $\pref_0$ (and $\pref_1$) are  essentially optimal when it  $p$ is close to~1. That is,  $\cort$ is  optimal when the system is highly unreliable, whereas algorithms like $\pref_0$ or $\pref_1$ are essentially optimal whenever the system is highly reliable. 
We start with the case of large $p$ value. 

Using a variant of the classical bivalence argument, we show
that in some cases the output depends on a single process,
and therefore we can prove the following result.

\begin{theorem}\label{thm:lower-bound-n-arbitrary-p-large}
For every $n\geq 2$ and $\epsilon\in (0,1]$, 
 for every 1-round consensus algorithm $\alg$ for $n$ processes in
the model with $q\le \epsilon/n$,  we
have $\Pr[\alg~\text{errs}]\ge(1-\epsilon)q$. 
\end{theorem}

\begin{proof}
Let $n\geq 2$, and let $\alg$ be a 1-round deterministic consensus algorithm for $n$ processes. If there exists an input configuration $\mathbf{x}\in\{0,1\}^n$ such that $\alg$ errs whenever all processes succeed to broadcast, then $\Pr[\alg\;\mbox{fails}\mid \mathbf{x}]\geq p^n$. The theorem follows by choosing $\hat p$ satisfying $\hat p^n\geq 1-\hat p$. From this point on, we are thus assuming that $\alg$ never fails when all processes succeed to broadcast. 

Let $\phi_0$ be the failure pattern where all messages are delivered, and, for $i\in\{1,\dots,n\}$, let $\phi_i$ be the failure pattern where all messages are delivered---except the message broadcasted by process~$i$. Let $\Phi=\{\phi_0,\dots,\phi_n\}$. Let us show that there exists $\mathbf{x}\in\{0,1\}^n$, and $\phi\in\Phi$ such that $\alg$ errs on input $\mathbf{x}$ with a failure pattern~$\phi$. We use the standard valency argument~\cite{AttiyaW04} for establishing this fact. Specifically, let us assume, for the purpose of contradiction, that, for every $\mathbf{x}\in\{0,1\}^n$, and every $\phi\in\Phi$, $\alg$ succeeds to solve consensus on input  $\mathbf{x}$ with failure pattern~$\phi$, with probability~1. 

We first claim that one of the $n+1$ input configurations $\mathbf{x}_0,\dots,\mathbf{x}_n$ is bivalent even for failure patterns restricted to be in~$\Phi$, where $\mathbf{x}_i=(1,\dots,1,0,\dots,0)$ with $i$ leading~1s, and $n-i$ following~0s, for $i\in\{0,\dots,n\}$. Indeed, if all the inputs $\mathbf{x}_0,\dots,\mathbf{x}_n$ were univalent, then let $i\geq 1$ be the smallest index such that $\mathbf{x}_i$ is 1-valent, and $\mathbf{x}_{i-1}$ is 0-valent. Such an index must exist as $\mathbf{x}_0$ is 0-valent, and $\mathbf{x}_n$ is 1-valent.  Let us consider $\alg$ on inputs $\mathbf{x}_{i-1}$ and $\mathbf{x}_i$. All processes but process~$i$ do not distinguish these two inputs whenever the failure pattern is~$\phi_i$. This is a contradiction, as they must output~0 for $\mathbf{x}_{i-1}$, and 1 for~$\mathbf{x}_i$.

Let $\mathbf{x}_i$ be a bivalent configuration for failure patterns in~$\Phi$. Let us assume, w.l.o.g., that $\alg$ outputs~0 on $\mathbf{x}_i$ with no failures (i.e., with~$\phi_0$).  Let $j\geq 1$ be the smallest index such that $\alg$ outputs~1 on $\mathbf{x}_i$ with the failure pattern~$\phi_j$. Process~$j$ does not distinguish $\phi_0$ from $\phi_j$ as it cannot tell whether its broadcast has succeeded. This yields a contradiction since it must output~0 on $\phi_0$, but 1 on $\phi_j$. So there exists $\mathbf{x}\in\{0,1\}^n$ and $\phi\in\Phi$ such that $\alg$ errs on $\mathbf{x}$ with pattern~$\phi$. 

It follows that $\alg$ fails with probability at least $\Pr(\phi)=p^{n-1}q$. The claim follows by choosing $\hat p=1-\epsilon/n$ because $(1-\epsilon/n)^{n-1}\ge (1-\epsilon/n)^{n}\geq 1-\epsilon$. 
\end{proof}

Note that the error probability approaches $q$ as $p$ approaches $1$.

Next, we consider the case of small $p$ values. Analyzing the reduced Kripke graph, we can prove the following tight result.

\begin{theorem}\label{thm:lower-bound-n-arbitrary-p-small}
    Let $n\geq 2$, and let $\hat{p}= \frac1n$. For every input $\mathbf{x}\in\{0,1\}^n$ with $d\in\{1,\dots,n-1\}$ entries equal to~1, and for every $p\leq \hat p$, every 1-round consensus algorithm errs on $\mathbf{x}$ with probability at least $$
    \sum_{i=0}^{\min\Set{d,n-d}} {d\choose i}{n-d\choose i}p^{2i}q^{n-2i}~.
    $$
\end{theorem}

\begin{proof}
Let $\mathbf{x}=(x_1,\dots,x_n)\in\{0,1\}^n$ with $d\in\{1,\dots,n-1\}$ entries equal to~1, where $1\leq d \leq n-1$. 
 Let us define $I_1=\{i\in [n]\mid x_i=1\}$, and let $I_0=[n]\smallsetminus I_1$.  So, $|I_1|=d$, and $|I_0|=n-d$. 
Assume, without loss of generality, that $d\le n-d$.
The proof is a generalization of the proof of Proposition~\ref{prop:lwb-simple}.

Let us first focus on the leading term $d(n-d) p^2q^{n-2}$ of the sum, which will also serve us as providing intuition for all the proof.
The hypercube $Q_\mathbf{x}$ shares a sub-hypercube of dimension $d$ with $Q_\mathbf{1}$ (i.e., the hypercube corresponding to the input $(1,\dots,1)$), and a sub-hypercube of dimension $n-d$ with $Q_\mathbf{0}$ (i.e., the hypercube corresponding to the input $(0,\dots,0)$). Let $\mathbf{w}_0=(0,\dots,0)$ denote the ``central'' vertex of the reduced Kripke graph~$\mathcal{Q}$, which corresponds to the  failure pattern $(\bot,\dots,\bot)$. Let $\mathbf{u}_1, \dots, \mathbf{u}_d$ be the $d$ neighbors of $\mathbf{w}_0$ belonging to $Q_\mathbf{x}\cap Q_\mathbf{1}$, and  let $\mathbf{v}_1, \dots, \mathbf{v}_{n-d}$ be the $n-d$ neighbors of $\mathbf{w}_0$ belonging to $Q_\mathbf{x}\cap Q_\mathbf{0}$. See Figure~\ref{fig:lower-bound-courteous-p-small-generalized}. Each $\mathbf{u}_i$ corresponds to a failure pattern $\varphi(\mathbf{u}_i)$ with a unique $\top$ entry appearing at some position in $I_1$, according to the function $\varphi$ defined in Eq.~\eqref{eq:delivery-pattern-associated-to-vertices}. Similarly, each $\mathbf{v}_i$ corresponds to a failure pattern $\varphi(\mathbf{v}_i)$ with a unique $\top$ entry appearing at some position in~$I_0$. 

\begin{figure}[htb]
\centering
\includegraphics[scale=.3]{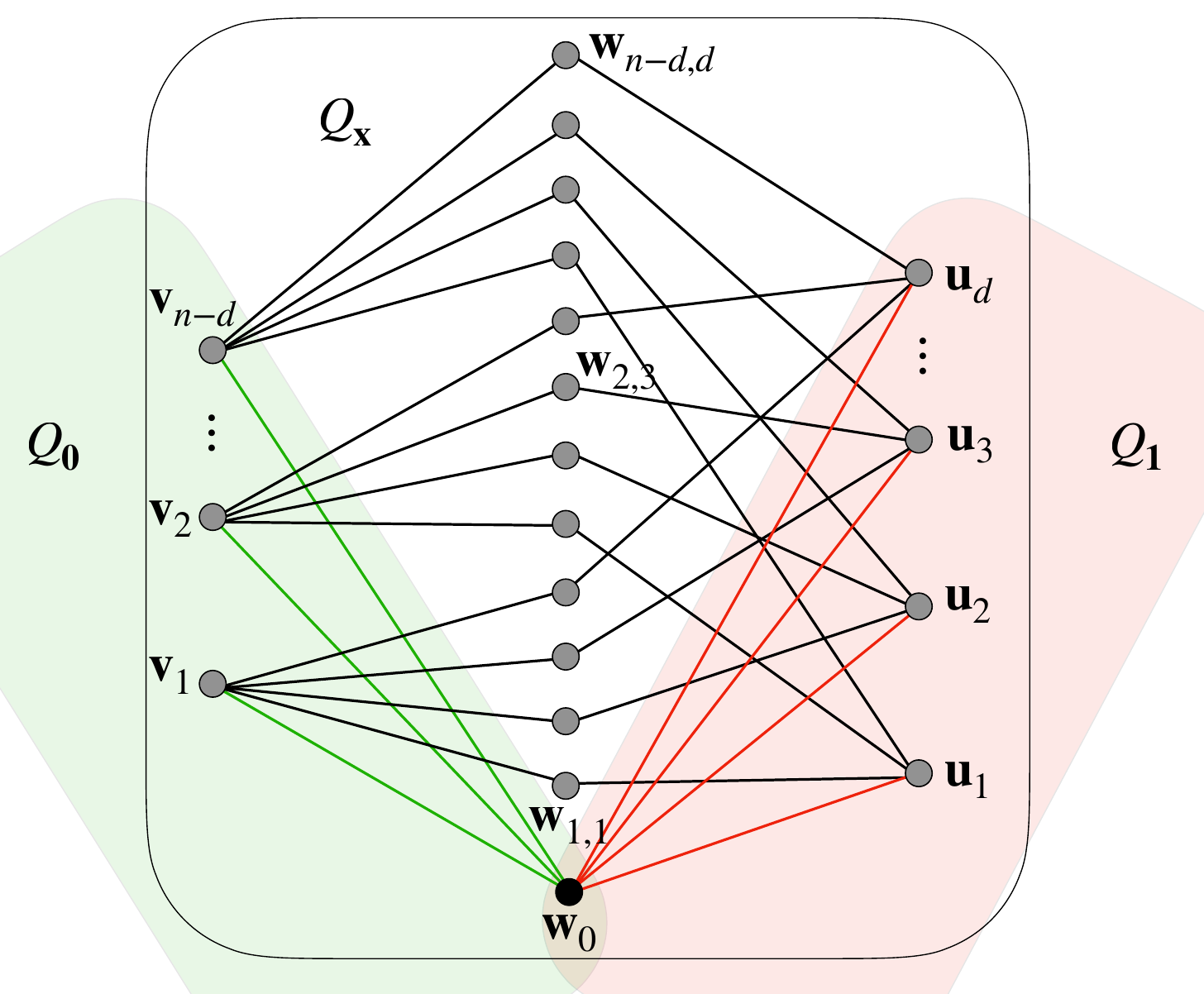}
\caption{$Q_\mathbf{x}$ and $Q_\mathbf{1}$ share a $d$-dimensional hypercube, while $Q_\mathbf{x}$ and $Q_\mathbf{0}$ share an $(n-d)$-dimensional hypercube. The neighbors of $\mathbf{w}_0$ in $Q_\mathbf{x}\cap Q_\mathbf{0}$ are $\mathbf{v}_1,\dots,\mathbf{v}_{n-d}$, and the neighbors of $\mathbf{w}_0$ in $Q_\mathbf{x}\cap Q_\mathbf{1}$ are $\mathbf{u}_1,\dots,\mathbf{u}_d$.}
\label{fig:lower-bound-courteous-p-small-generalized}
\end{figure}

For every $i\in\{1,\dots,n-d\}$, and every $j\in\{1,\dots,d\}$, there exists a vertex $\mathbf{w}_{i,j}\in Q_\mathbf{x}$, distinct from $\mathbf{w}_0$, such that $P_{i,j}=(\mathbf{v}_i,\mathbf{w}_{i,j},\mathbf{u}_j)$ is a path of length~2 between $\mathbf{v}_i$ and $\mathbf{u}_j$ (see Figure~\ref{fig:lower-bound-courteous-p-small-generalized}). Specifically, if $\varphi(\mathbf{v}_i)$ has its unique $\top$-entry at position~$k_i\in I_0$, and $\varphi(\mathbf{u}_j)$ has its unique $\top$-entry at position~$k_j\in I_1$, then $\varphi(\mathbf{w}_{i,j})$ has  exactly two$\top$-entries, one at position~$k_i$ and the other at position~$k_j$. There are therefore $d(n-d)$ distinct vertices~$\mathbf{w}_{i,j}$. Moreover, for every $(i,j)$, $\Pr[\mathbf{w}_{i,j}]=p^2q^{n-2}$. 

Let $\alg$ be a 1-round consensus algorithm, inducing  an edge-coloring $c$ of $\mathcal{Q}$. By  the validity condition, for every $i\in\{1,\dots,n-d\}$, the edge $\{\mathbf{w}_0,\mathbf{v}_i\}$ must be colored~0 by~$c$. Similarly, for every $j\in\{1,\dots,d\}$, the edge $\{\mathbf{w}_0,\mathbf{u}_j\}$ must be colored~1 by~$c$.
By Lemma~\ref{lem:each-path-has-a-cut}, each of the paths $P_{i,j}$ must therefore contain a vertex in $\cut_c(Q_\mathbf{x})$. We consider two cases. 

\begin{itemize}
    \item If none of the $n$ vertices 
    $\Set{\mathbf{v}_i}_{i=1}^{n-d}\cup\Set{\mathbf{u}_j}_{j=1}^d$, are in $\cut_c(Q_\mathbf{x})$ then the $d(n-d)$  vertices~$\mathbf{w}_{i,j}$ must be in the cut. As a consequence, $\alg$ fails on $\mathbf{x}$ with probability at least $\sum_{i,j}\Pr[\mathbf{w}_{i,j}]=d(n-d)p^2q^{n-2}$, and the theorem follows. 
    
    \item If there exists $i$ such that $\mathbf{v}_i\in \cut_c(Q_\mathbf{x})$, or if there exists $j$ such that such that $\mathbf{u}_j\in \cut_c(Q_\mathbf{x})$, then $\alg$ fails on $\mathbf{x}$ with probability at least $pq^{n-1}$, because $\Pr[\mathbf{v}_i]=\Pr[\mathbf{u}_j]=pq^{n-1}$ for every $i$ and~$j$.  For $p$ small enough, we have that $pq^{n-1}\geq d(n-d) p^2q^{n-2}$. 
\end{itemize}
In both cases, $\alg$ fails on $\mathbf{x}$ with probability at least $d(n-d) p^2q^{n-2}$.  

This approach can be applied not only to the first term of the sum $\sum_{i=1}^d {d\choose i}{n-d\choose i}p^{2i}q^{n-2i}$, but to all terms simultaneously, by the same arguments as above. Indeed, every vertex $\mathbf{v}\in Q_{\mathbf{0}}$ is associated with a delivery pattern $\phi=\varphi(\mathbf{v})$. For every $1\leq i \leq \min\{d,n-d\}$, let us define the following sets: 
\begin{itemize}
    \item $V_i=\Set{\mathbf{v}\in Q_\mathbf{0}\cap Q_\mathbf{x}\mid \mbox{the number of $\top$ symbols in $\varphi(\mathbf{v})$ is $i$}}$, and 
    
    \item $U_j=\Set{\mathbf{u}\in Q_\mathbf{1}\cap Q_\mathbf{x}\mid \mbox{the number of $\top$ symbols in $\varphi(\mathbf{u})$ is $i$}}$
\end{itemize} 
Note that the $\top$ symbols of the delivery patterns $\varphi(\mathbf{v})$ for $\mathbf{v}\in V_i$ appear at positions in~$I_0$, whereas the $\top$ symbols of the delivery patterns $\varphi(\mathbf{u})$ for $\mathbf{u}\in U_i$ appear at positions in~$I_1$. Let 
\begin{align*}
W_j=\{\mathbf{w}\in Q_\mathbf{x}\mid~ & \mbox{$\varphi(\mathbf{w})$ has exactly $j$ $\top$ symbols at positions in $I_0$,}\\
& \mbox{and exactly $j$ $\top$ symbols at positions in $I_1$}\}.
\end{align*}
For every $\mathbf{v} \in V_i$ and every $\mathbf{u} \in U_i$ there exists
a path of length $2i$ connecting them through a vertex $\mathbf{w}\in W_i$. This path is the following. 
Let $J_0\subseteq I_0$ be the set of indices at which the $\top$ symbols of $\varphi(\mathbf{v})$ occur, and let $J_1\subseteq I_1$ be the set of indices at which the $\top$ symbols of $\varphi(\mathbf{u})$ occur. The path starts at~$\mathbf{v}$, then the indices of $J_1$ are turned to $\top$ one by one in arbitrary order until $\mathbf{w}$ is reached, and then the indices of $J_0$
are turned to $\bot$ one by one in arbitrary order until $\mathbf{u}$ is reached.
Each vertex on such a path is contained in $Q_\mathbf{x}$, and has probability $p^{j+k}q^{n-j-k}$ for some
$0\le k\le j$. 

Let $\mathbf{S}_i$ be the set of paths constructed this way. Lemma~\ref{lem:each-path-has-a-cut} says that each of the paths in $\mathbf{S}_i$ must  contain a vertex in $\cut_c(Q_\mathbf{x})$. Since we have $|V_i|={n-d\choose i}$, and $|U_i|={d\choose i}$, we have $|\mathbf{S}_i|=|W_i|={d\choose i}{n-d\choose i}$. Algorithm $\cort$ is precisely placing in the cut all the vertices in~$W_i$, and the question is, can we do better? A crucial observation is that, for every $1\leq i<j \leq \min\{d,n-d\}$, and for every $P_i\in \mathbf{S}_i$ and $P_j\in \mathbf{S}_j$, we have $P_i\cap P_j=\varnothing$. As a consequence, one must treat all the paths in $\mathbf{S}_i$ independently from all the paths in $\mathbf{S}_j$. That is, a cut vertex for a path in $\mathbf{S}_i$ does not cut any path $\mathbf{S}_j$. In other words, we can treat each $i$ independently. 

Let $\alg$ be a 1-round consensus algorithm, inducing  an edge-coloring $c$ of $\mathcal{Q}$. We consider two cases. 
\begin{itemize}
    \item If all the vertices in $\bigcup_{1\leq i \leq \min\{d,n-d\}}W_i$ are in the cut, then $\alg$ fails on $\mathbf{x}$ with probability at least $\sum_{i=1}^d {d\choose i}{n-d\choose i}p^{2i}q^{n-2i}$ and the theorem follows. 
    
    \item If there exists $i$ such that at least one vertex in $W_i$ is not in the cut, then there must be another vertex $\mathbf{z}\in Q_\mathbf{x}$  in the cut, with $\mathbf{z}\notin W_i$. Such a vertex $\mathbf{z}$ has probability $p^{i+k}q^{n-i-k}$ for some $0\le k < i$. As such, it cuts at most, in fact, exactly ${n-d-k \choose i-k}$ paths of $\mathbf{S}_i$. However, for $p$ small enough, 
    \[
    p^{i+k}q^{n-i-k}\geq {n-d-k \choose i-k}p^{2i}q^{n-2i}~,
    \]
     because  whenever
    $p\le \frac1n$,
     $${
     n-d-k \choose i-k}\leq 
    (n-d)^{i-k}\le
    \left(\frac{q}{p}\right)^{i-k}$$ 
    where the second inequality follows from the definition of $q$ and from the fact that $d\ge1$.
\end{itemize}
It follows that, for $p\leq \hat{p}= \frac1n%\frac{1}{en+1}
$, every 1-round algorithm for consensus fails with probability at least $\sum_{i=1}^d {d\choose i}{n-d\choose i}p^{2i}q^{n-2i}$.
\end{proof}

Comparing the above lower bound with \lemmaref{thm:court-prob}, we 
deduce that Algorithm $\court$ is optimal for all $p\le 1/n$.
In addition,
Theorem~\ref{thm:lower-bound-n-arbitrary-p-small} implies a remarkable property of $\cort$. 

\begin{corollary}
For every $n\geq 2$ and $p\le \nicefrac1n$, 
$\cort$ is an optimal 1-round consensus algorithm for $n$ processes 
even if the input configurations are restricted to be one of the three instances $\{0^n,\mathbf{x},1^n\}$ for any fixed
$\mathbf{x}\in\{0,1\}^n\smallsetminus\Set{0^n,1^n}$.
\end{corollary}

%%%%%%%%%%%%%%%%%%%%%%%%%%%%%%%%%
\section{Conclusion and Open Problems}
\label{sec:conc}
%%%%%%%%%%%%%%%%%%%%%%%%%%%%%%%%%

In this paper, we explored the error probability of consensus protocols
in the stochastic broadcast model. 
Our study leaves  many interesting questions open. 
The basic concrete question is to determine, for any given $r,n$, and  $p\in[0,1]$, the optimal error probability. While we have pointed
out a few candidate algorithms, we were able to prove optimality only
for some choices of the parameters. 
As we showed, the problem can be reduced to determining the minimum weight
edge cut of certain simplicial complexes, or alternatively, the minimum weight vertex cut  of  certain  hypercubic structures related to the Kripke graph.

It would be interesting to consider problems other than consensus, especially set agreement,  where it is no longer true that the error 
probability can be characterized by  minimum weight cuts of simplicial complexes. Techniques about generalized max-flows and min-cuts in simplicial
complexes might be relevant~\cite{MaxwellN21}. 
Approximate agreement has been studied, but only for two processes~\cite{FraigniaudPSR25}. It was shown that, roughly speaking, the one round error probability of solving approximate agreement is the same as for solving consensus  (for any value of $p$),
 while for three processes this is no longer the case.
 The optimal error probability of approximate agreement  error is not known for $n\geq 3$, even in the one-round case.
 
In this paper, we looked at the stochastic \emph{broadcast} model. 
Another natural generalization of the stochastic model proposed 
in~\cite{FraigniaudPSR25} is the stochastic \emph{unicast} model, in which there is a channel in each direction between every two processes, and 
each message may be lost  with some given probability $0\le p\le 1$,
independently. Another interesting variant is
where the channels are undirected; namely, in each round, a channel either delivers both messages 
(one in each direction) with probability $p$, or drops
both messages with probability $q=1-p$. Both versions, as well as others, have been considered in the past (e.g., in~\cite{elhayekDisc2024,KuhnO11,MichailSpirakis2018,SWK2009}), but not for consensus. %, as we do here. 

Results presented in this paper show a particularly pleasing
structure; it is possible that nice structures can be 
discovered for other related problems.
We believe that work in this direction will improve our
understanding of randomized distributed computation in general.

%%%%%%%%%%%%%%%%%%%%%%%%%%%%
\bibliographystyle{plain}
\bibliography{agreement}
%%%%%%%%%%%%%%%%%%%%%%%%%%%%%

%%%%%%%%%%%%%%%%%%%%%%%%%%%%%
\vspace{1cm}
\appendix
\centerline{\Large\bf A P P E N D I X}
%%%%%%%%%%%%%%%%%%%%%%%%%%%%%%

%%%%%%%%%%%%%%%%%%%%%%%%%%%%%%
\section{The Protocol Complex of the Stochastic Broadcast Model}
\label{sec:Protocol-Complex}
%%%%%%%%%%%%%%%%%%%%%%%%%%%%%%%

The goal of this appendix is to show the tight connections between the topological approach of distributed computing (i.e., the analysis of the structure of the protocol complex, and its ability to be mapped via a simplicial map to the output complex), and the epistemic approach of distributed computing (i.e., the analysis of the knowledge acquired by the agents, and of the structure of the associated Kripke graph), in the context of the stochastic broadcast model. 

To understand the structure of the Kripke graph of consensus after one round in the stochastic broadcast model,  
consider its 1-round \emph{protocol complex}~\cite{HerlihyKR13}. The vertices of this complex are individual local states of processes after one round. Each vertex is labeled with a process name and its local state. A set of vertices labeled with different process names is a face (a.k.a.~simplex) of the protocol complex if the corresponding collection of states of these processes may globally result from a valid execution. The \emph{facets} of the complex are the faces of cardinality~$n$, i.e., dimension $n-1$. 

\begin{wrapfigure}{r}{3cm}
\centering
\includegraphics[scale=.4]{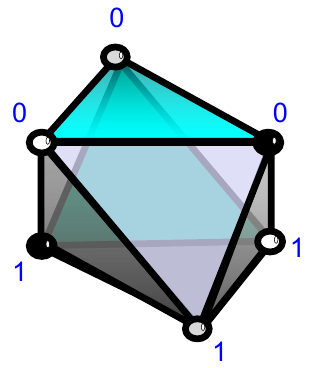}
\caption{Complex for three processes and binary inputs.}
\label{fig:inputComplex}
\end{wrapfigure}

For both the Kripke graph and the protocol complex, we may assume without loss of generality that the algorithm is \emph{full information}, that is, each process remembers its full history, and sends in each message its entire local state. In case of a single round, this means that each process sends its name and its input. 

The \emph{input complex} $\mathcal{I}$ for binary consensus is the protocol complex after $r=0$ rounds. A local state of a process $i\in [n]$ is a pair $(i,x_i)$ with $x_i\in\{0,1\}$ is the input of process~$i$. Since all combinations of binary inputs are possible among the $n$ processes, the facets of $\mathcal{I}$ can also be viewed as all the vectors in $\{0,1\}^n$. $\mathcal{I}$ is therefore the so-called \emph{pseudo-sphere} $\Psi([n],\{0,1\})$~\cite{HerlihyKR13}. 
See Figure~\ref{fig:inputComplex} for a geometric representation of $\mathcal{I}$ for $n=3$ processes.

\subsection{The Protocol Complex after One Round}

After one round, the state of a process $i\in [n]$ is $\mathbf{s}_i=\{(j,x_j)\mid j\in R_i\}$ of pairs $(j,x_j)$ received during the round; that is, $\{i\}\subseteq R_i \subseteq [n]$, and, for every $j\in R_i$, there is a unique $x_j\in\{0,1\}$ such that $(j,x_j)\in \mathbf{s}_i$. It is more convenient to encode such a state as a vector $\mathbf{s}_i=(\mathbf{s}_i[1],\dots,\mathbf{s}_i[n])\in \{0,1,\star\}^n$ where $\mathbf{s}_i[j]=\star$ means that process~$i$ has not received any message from process~$j$ during the round.  A facet of the protocol complex $\mathcal{P}$ is then a vector $\bar{\mathbf{s}}=(\mathbf{s}_1,\dots,\mathbf{s}_n)\in (\{0,1,\star\}^n)^n$ satisfying the following broadcast condition: 
\[
\forall i\in [n], \forall j,j'\in [n]\smallsetminus \{i\}: \;
(\mathbf{s}_j[i]=\mathbf{s}_{j'}[i]) \land (\mathbf{s}_j[i]\neq \star \Rightarrow \mathbf{s}_j[i]=\mathbf{s}_i[i]) \land (\mathbf{s}_i[i]=x_i),
\]
where $x_i$ denotes the input of process~$i$.
This condition captures the broadcast communication, which may or may not succeed, but if a process $j$ has received a message from another process~$i$, then every other process $j'$ has also received the same message from process~$i$. 

\subparagraph{Example.}

Figure~\ref{fig:simplex-flat-proba-mainText} displays the subcomplex $\mathcal{P}(\mathbf{x})$ of $\mathcal{P}$ for $n=3$ where the three processes $i\in \{1,2,3\}$ have respective input $x_1,x_2,x_3$, i.e., $\mathbf{x}=(x_1,x_2,x_3)$ is the input configuration. Each facet of the complex $\mathcal{P}(\mathbf{x})$ is represented as a full triangle. There are eight such triangles, each one representing a possible state of the system after one round, starting from the global state~$\mathbf{x}$. 
\medskip

In general, for an $n$-process system, the subcomplex $\mathcal{P}(\mathbf{x})$ of $\mathcal{P}$ induced by the input configuration $\mathbf{x}=(x_1,\dots,x_n)\in\{0,1\}^n$ has $2^n$ facets. There is a one-to-one correspondence between the facets of the subcomplex $\mathcal{P}(\mathbf{x})$ and the $2^n$ possible delivery patterns $\phi=(\phi_1,\dots,\phi_n)\in\{\bot,\top\}^n$, where $\phi_i=\bot$ represents the fact that the broadcast of process~$i$ fails, and $\phi_i=\top$ represents the fact that the broadcast of process~$i$ succeeds. In total, the protocol complex $\mathcal{P}$ has $2^n\cdot 2^n$ facets~$\bar{\mathbf{s}}$, one for each pair $(\mathbf{x},\phi)$ where $\mathbf{x}=(x_1,\dots,x_n)\in\{0,1\}^n$ is an input configuration, and $\phi=(\phi_1,\dots,\phi_n)\in\{\bot,\top\}^n$ is a delivery pattern. For any input configuration~$\mathbf{x}$, and for any delivery pattern~$\phi$, the probability $\Pr[\bar{\mathbf{s}}]$ of the facet $\bar{\mathbf{s}}=(\mathbf{x},\phi)$ is merely $\Pr[\phi]=p^d q^{n-d}$ where $d$ denotes the number of occurrences of $\top$ in~$\phi$. 

%%%%%%%%%%%%%%%%%%%%%%%%%%%%%%%%%
\subsection{Algorithms as Simplicial Maps}
%%%%%%%%%%%%%%%%%%%%%%%%%%%%%%%%%

The protocol complex representations is valuable in particular because, when solving consensus (or any other task), the decisions of the processes in a 1-round algorithm define a map $f$ from the vertices of the protocol complex $\mathcal{P}$ to the vertices of the \emph{output complex}~$\mathcal{O}$. (More generally, for every $r\geq 0$, the decisions of the processes in an $r$-round algorithm define a map $f$ from the vertices of the protocol complex $\mathcal{P}^r$ to the vertices of~$\mathcal{O}$.) In the case of the consensus task,  $\mathcal{O}$~has two facets: $\mathbf{y}_0=(0,\dots,0)$ and $\mathbf{y}_1=(1,\dots,1)$. For the algorithm to be correct, the corresponding map $f:V(\mathcal{P})\to V(\mathcal{O})$ must satisfy two conditions. 
\begin{enumerate}
    \item It must be \emph{simplicial}, that is, it must preserve simplexes, as the image of a configuration reached by the system after $r$ rounds must be a legal output configuration. 
    
    \item It must \emph{agree with the specification} of the task. This means that for any input configuration~$\mathbf{x}=(x_1,\dots,x_n)$, and for any configuration $\bar{\mathbf{s}}=(\mathbf{s}_1,\dots,\mathbf{s}_n)\in\mathcal{P}(x)$, i.e., any configuration  that may be reached by the system after one round starting from~$\mathbf{x}$, the output $f(\bar{\mathbf{s}})=(f(\mathbf{s}_1),\dots,f(\mathbf{s}_n))$ must belong to the set $\Delta(\mathbf{x})$ of simplices of~$\mathcal{O}$ that are specified as legal for~$\mathbf{x}$. 
\end{enumerate}
See Figure~\ref{fig:optAlgsComplex-mainText}
and compare with Figure~\ref{fig-pref-courteous-hyperplane}.
For instance, in the case of consensus, for the input configuration $\mathbf{x}_b=(b,\dots,b)=b^n$ with $b\in\{0,1\}$,  one must have $f(\bar{\mathbf{s}})=(b,\dots,b)$. That is, for every $b\in\{0,1\}$, $\Delta(\mathbf{x}_b)=\mathbf{y}_b$. Instead, for every $\mathbf{x}\in\{0,1\}^n\smallsetminus\{\mathbf{x}_0,\mathbf{x}_1\}$, $\Delta(\mathbf{x})=\mathcal{O}=\{\mathbf{y}_0,\mathbf{y}_1\}$.

%%%%%%%%%%%%%%%%%%%%%%%%%%%%%%%%%
\subsection{Minimizing the Error Probability}
%%%%%%%%%%%%%%%%%%%%%%%%%%%%%%%%%

In the case of the stochastic broadcast model, the output of a 1-round algorithm may or may not be correct depending on the messages delivered during its execution, which is random. As a consequence, even if the output of process is a deterministic function of its local state after one round, this local state is a random function of the inputs of the $n$ processes. Hence, we are interested in minimizing the error probability of the algorithm~$f$, that is, we aim at minimizing 
\[
\max_{\mathbf{x}\in \{0,1\}^n}
\Pr[\{\bar{\mathbf{s}}\in\mathcal{P}(\mathbf{x})\mid f(\bar{\mathbf{s}})\notin \Delta(\mathbf{x})\}].
\]
The maximization is performed over all the facets $\mathbf{x}$ of $\mathcal{I}$. Indeed, there are no process failures, but just stochastic failures of message deliveries. Therefore, each and every process must produce an output. Note that 
\[
\Pr[\{\bar{\mathbf{s}}\in\mathcal{P}(\mathbf{x})\mid f(\bar{\mathbf{s}})\notin \Delta(\mathbf{x})\}]=
\sum_{\bar{\mathbf{s}}=(\mathbf{x},\phi)\in\mathcal{P}(\mathbf{x})\;\mid\;  f(\bar{\mathbf{s}})\notin \Delta(\mathbf{x})}\Pr[\phi].
\]
Moreover, as we restrict our attention to algorithms that systematically satisfy the validity condition, we are focusing on functions $f:V(\mathcal{P})\to\{0,1\}$ satisfying $f(\mathbf{s})=b$ for every $\mathbf{s}\in\{b,\star\}^n$. Let us denote $F_{val}$ this set of functions. The problem can thus be rephrased as 
\begin{equation}\label{eq:problem-expressed-using-topology}
\min_{f\in F_{val}}\;\;
\max_{\mathbf{x}\in \{0,1\}^n\smallsetminus \{0^n,1^n\}}
\;\; 
\sum_{\bar{\mathbf{s}}=(\mathbf{x},\phi)\in\mathcal{P}(\mathbf{x})
\;\mid\; f(\bar{\mathbf{s}})\notin \{0^n,1^n\}}\Pr[\phi].
\end{equation}
The problem can be formulated in a way similar for $r$-round algorithms, $r\geq 1$, but it already appears  uneasy to handle even for $r=1$. Yet, the topological approach sheds some light on the problem. In particular, it is worth making some observations that will prove useful for solving the problem of Eq.~\eqref{eq:problem-expressed-using-topology} via another route. 

%%%%%%%%%%%%%%%%%%%%%%%%%%%%%%%%%
\subsection{The Connected Components of the Protocol Complexes}
%%%%%%%%%%%%%%%%%%%%%%%%%%%%%%%%%

As mentioned before, each facet of $\mathcal{P}$ is a pair $(\mathbf{x},\phi)$, with $\mathbf{x}\in\{0,1\}^n$ and $\phi\in\{\bot,\top\}^n$. For $\mathbf{x}\in\{0,1\}^n$, we use $\mathcal{P}(\mathbf{x})$ to denote the subcomplex of $\mathcal{P}$ reachable from $\mathbf{x}$ after one round, i.e., the subcomplex with facets in $\{(\mathbf{x},\phi)\mid \phi\in\{\bot,\top\}^n\}$. Let us now focus on another subcomplex of $\mathcal{P}$ obtained by fixing a delivery pattern. 

\begin{definition}
    For any  delivery pattern $\phi\in\{\bot,\top\}^n$, let $\mathcal{P}(\phi)$ be the sub-complex of $\mathcal{P}$ induced by the  facets in $\{(\mathbf{x},\phi)\mid \mathbf{x}\in\{0,1\}^n\}$. 
\end{definition}

Note that $\mathcal{P}=\bigcup_{\phi\in\{\bot,\top\}^n} \mathcal{P}(\phi)$. Importantly, depending on $\phi$, $\mathcal{P}(\phi)$ may or may not be path-connected.%
\footnote{
In high-dimensional complexes, which generalize graphs to higher dimensions, connectivity may have different topological meaning. Path-connectivity, a.k.a.~$0$-connectivity, corresponds to the standard notion of connectivity as in graphs, i.e., the existence of a sequence of edges forming a path between two vertices.
}
For instance, if $\phi=\top^n$, then any two facets $(\mathbf{x},\phi)$ and $(\mathbf{x}',\phi)$ where $\mathbf{x}\neq \mathbf{x}'$ have no vertices in common. 
In contrast, if $\phi=\bot^n$, then  $\mathcal{P}(\phi)$ is path-connected, as any two facets $(\mathbf{x},\phi)$ and $(\mathbf{x}',\phi)$ intersect, except when $\mathbf{x}=(x_1,\dots,x_n)$ and $\mathbf{x}'=(x'_1,\dots,x'_n)$ with $x'_i=1-x_i$ for every $i\in [n]$. Thus, $\mathcal{P}(\bot^n)$ has a unique path-connected component.

\subparagraph{Notations.} 

We denote by $\texttt{CC}(\mathcal{P}(\phi))$ the set of path-connected components of~$\mathcal{P}(\phi)$ (or, equivalently, the set of connected components of the 1-dimensional skeleton $\mathsf{skel}_1(\mathcal{P}(\phi))$ of $\mathcal{P}(\phi)$). 
\smallbreak

If $d$ denotes the number of occurrences of $\top$ in~$\phi$, then $\mathcal{P}(\phi)$ has $2^d$ path-connected components, i.e., $|\texttt{CC}(\mathcal{P}(\phi))|=2^d$. Moreover, each component $C\in \texttt{CC}(\mathcal{P}(\phi))$ is composed of $2^{n-d}$ facets, with a common intersection equal to a simplex of cardinality~$d$. 
A remarkable property of the protocol subcomplex  in the stochastic broadcast model is the following  (see Figure~\ref{fig:simplex-flat-proba-mainText}). 

\begin{lemma}\label{lem:vertex-belongs-to-two-facets}
The following two properties are satisfied.
\begin{itemize}
    \item For any two distinct delivery patterns $\phi$ and $\phi'$, any two components $C\in\texttt{CC}(\mathcal{P}(\phi))$ and $C'\in\texttt{CC}(\mathcal{P}[\phi'])$ intersect in at most one vertex. 
    \item For every vertex $\mathbf{s}$ of $\mathcal{P}$, there exist exactly two distinct delivery patterns $\phi$ and $\phi'$ such that $\mathbf{s}=C\cap C'$ with $C\in\texttt{CC}(\mathcal{P}(\phi))$ and $C'\in\texttt{CC}(\mathcal{P}(\phi'))$.
\end{itemize}
\end{lemma}

\begin{proof}
Let $\phi$ and $\phi'$ be two distinct delivery patterns, and let $C$ and $C'$ be two path-connected components of $\mathcal{P}(\phi)$ and $\mathcal{P}(\phi')$, respectively. If the Hamming distance of $\phi$ and $\phi'$ is larger than~1, then $C\cap C'=\varnothing$. Indeed, for every process~$i$, there exists at least another process $j\neq i$ such that $i$ receives a message from $j$ in all executions corresponding to one of the two components, but does not receive any message from $j$ in all the executions corresponding to the other. 
If the Hamming distance of $\phi=(\phi_1,\dots,\phi_n)$ and $\phi'=(\phi'_1,\dots,\phi'_n)$ is exactly~1, then let $i\in[n]$ be the index for which $\phi_i\neq \phi'_i$, while $\phi_j= \phi'_j$ for all $j\neq i$. All processes $j\neq i$ distinguish $C$ from $C'$ as they receive from $i$ in one, whereas they do not receive for $i$ in the other. However, process $i$ itself may belong to both $C$ and $C'$. This is the case if and only if  these components respectively correspond to inputs $\mathbf{x}$ and $\mathbf{x}'$ where $x_\ell=x'_\ell$ for every $\ell\in\{i\}\cup \{j\in[n] \mid \phi_j=\phi'_j=\top\}$. Indeed, under such circumstances, process $i$ cannot distinguish $\phi$ from $\phi'$ based on the received data, nor can it distinguish the component in which its own message was delivered from the one in which it was not.    

Conversely, let $\mathbf{s}=(i,s_i)$ be a vertex of  $\mathcal{P}$, and let $(\mathbf{x},\phi)$ be a facet containing~$\mathbf{s}$. Let $C$ be the path-connected component of $\mathcal{P}(\phi)$ containing $(\mathbf{x},\phi)$. If $\phi_i=\top$ (resp., $\phi_i=\bot$), then let $\phi'$ be the delivery pattern identical to $\phi$ but the $i$th entry, for which $\phi'_i=\bot$ (resp., $\phi'_i=\top$). We have $\mathbf{s}\in(\mathbf{x},\phi')$ as process $i$ cannot distinguish $\phi'$ from $\phi$, and more generally, $\mathbf{s}\in C\cap C'$ where $C'$ is the path-connected component of $\mathcal{P}(\phi')$ containing $(\mathbf{x},\phi')$.  
\end{proof}

%%%%%%%%%%%%%%%%%%%%%%%%%%%%%%%%%
\subsection{The Intersection Graph of the Protocol Complex}
%%%%%%%%%%%%%%%%%%%%%%%%%%%%%%%%%

Lemma~\ref{lem:vertex-belongs-to-two-facets} suggests that, instead of looking for a simplicial map $f:V(\mathcal{P})\to \{0,1\}$ solution of Eq.~\eqref{eq:problem-expressed-using-topology}, one could simply look for a function $f':E(G_{\mathcal{P}}) \to \{0,1\}$ where $G_{\mathcal{P}}$ is the intersection graph, defined as 
\[
V(G_{\mathcal{P}})=\bigcup_{\phi\in\{\bot,\top\}^n}\mathtt{CC}(\mathcal{P}(\phi)),
\;\;\;\text{and}\;\;\;
E(G_{\mathcal{P}})=\big\{\{C,C'\}\in V(G)\times V(G)\mid C\cap C'\neq\varnothing\big\}.
\]
That is, each vertex  of $G_{\mathcal{P}}$ is a path-connected component of $\mathcal{P}(\phi)$ for some delivery pattern~$\phi$, and there is an edge between two distinct components if and only if they intersect.  
The  remaining issues are:
\begin{enumerate}
    \item understanding the structure of the intersection graph~$G_{\mathcal{P}}$, and
    \item rephrasing the minimization problem stated in Eq.~\eqref{eq:problem-expressed-using-topology} in terms of properties to be satisfied by an edge-coloring of~$G_{\mathcal{P}}$. 
\end{enumerate}

%%%%%%%%%%%%%%%%%%%%%%%%%%%%%%%%%
\subparagraph*{Take Away Message.}
%%%%%%%%%%%%%%%%%%%%%%%%%%%%%%%%%

These series of observations about the analysis of the stochastic broadcast model using the topological approach motivated our interest for the Kripke graph. Lemma~\ref{lem:vertex-belongs-to-two-facets}  actually suggests that one may even consider a simplified variant of the Kripke graph, resulting from merging some of its vertices into one as long as (1)~these vertices are of the form $(\mathbf{x},\phi)$ for the same delivery pattern~$\phi$, and (2)~they represent configurations that are indistinguishable by at least one process. This is the approach that we developed in the core part of the paper. 

%%%%%%%%%%%%%%%%%%%%%%%%%%%%%%%%%
\end{document}